\documentclass[11pt]{article}
\pdfoutput=1 

\usepackage[sc]{mathpazo}

\usepackage{palatino,microtype,dsfont}
\usepackage[margin=15pt,font=small,labelfont={bf}]{caption}
\usepackage[margin=1in]{geometry}
\usepackage[english]{babel}
\usepackage[utf8]{inputenc}
\usepackage[compact]{titlesec}
\usepackage{cmap}
\usepackage[T1]{fontenc}
\usepackage{bm}
\usepackage{tikz}
\usepackage{colortbl,booktabs} 
\usepackage{mathtools}
\usepackage{amsfonts}
\usepackage{amsmath}
\usepackage{amssymb}
\usepackage{amsthm}
\usepackage{float}
\usepackage{graphics}
\usepackage[numbers,sort]{natbib}
\usepackage[backref=page]{hyperref}
\usepackage[svgnames]{xcolor}
\hypersetup{colorlinks={true},urlcolor={blue},linkcolor={DarkBlue},citecolor=[named]{DarkGreen},linktoc=all}
\renewcommand*{\backref}[1]{}
\renewcommand*{\backrefalt}[4]{%
    \ifcase #1 (Not cited.)%
    \or        (Cited on page~#2)%
    \else      (Cited on pages~#2)%
    \fi}
\usepackage{framed}
\usepackage[capitalise,nameinlink,noabbrev]{cleveref}
\usepackage{doi}
\usepackage{fontawesome5}

\usepackage{paralist}

\newcommand{\BibTeX}{\rm B\kern-.05em{\sc i\kern-.025em b}\kern-.08em\TeX}
\usepackage[justification = centering]{subcaption}
\usepackage{tabularx}
\usepackage{thmtools,thm-restate}
\usepackage[linesnumbered, ruled,vlined]{algorithm2e}
\SetKwComment{Comment}{/* }{ */}

\newtheorem{definition}{Definition}[section]
\newtheorem{example}{Example}[section]
\newtheorem{theorem}{Theorem}[section]
\newtheorem{lemma}[theorem]{Lemma}%

\newtheorem{remark}{Remark}[section]
\newtheorem{observation}{Observation}[section]

\newtheorem{proposition}[theorem]{Proposition}

\usepackage[capitalise,nameinlink,noabbrev]{cleveref}
\Crefname{claim}{Claim}{Claims}
\Crefname{corollary}{Corollary}{Corollaries}
\Crefname{definition}{Definition}{Definitions}
\Crefname{example}{Example}{Examples}
\Crefname{lemma}{Lemma}{Lemmas}
\Crefname{property}{Property}{Properties}
\Crefname{proposition}{Proposition}{Propositions}
\Crefname{remark}{Remark}{Remarks}
\Crefname{theorem}{Theorem}{Theorems}

\crefname{algocf}{Algorithm}{Algorithms}
\Crefname{claim}{Claim}{Claims}
\Crefname{corollary}{Corollary}{Corollaries}
\Crefname{definition}{Definition}{Definitions}
\Crefname{example}{Example}{Examples}
\Crefname{lemma}{Lemma}{Lemmas}
\Crefname{proposition}{Proposition}{Propositions}
\Crefname{remark}{Remark}{Remarks}
\Crefname{theorem}{Theorem}{Theorems}

\DeclareMathOperator*{\argmax}{arg\,max}
\DeclareMathOperator*{\argmin}{arg\,min}
\newcommand{\ess}{\mathcal{S}}

\colorlet{myred}{red!25}
\colorlet{myblue}{blue!25}
\colorlet{mygreen}{green!25}
\colorlet{mygrey}{gray!15}

\usepackage{multirow,array}
\newcommand{\Instance}{\textup{\textsc{Menu Selection}}}

\usepackage[textwidth=1.2cm]{todonotes}

\title{Menu Selection:\\ A Computational Approach to Minimizing Food Waste}

\author{
\begin{tabular}{m{0.25\linewidth}m{0.2\linewidth}m{0.2\linewidth}m{0.25\linewidth}}
		  \multicolumn{2}{c}{\textbf{Haris Aziz}}                                                                       & \multicolumn{2}{c}{\textbf{Nicholas Mattei}}\\
		\multicolumn{2}{c}{\small{UNSW Sydney, Australia}}                                                            & \multicolumn{2}{c}{\small{Tulane University, USA}}\\
		\multicolumn{2}{c}{\href{mailto:haris.aziz@unsw.edu.au}{\small{\texttt{haris.aziz@unsw.edu.au}}}\quad}        & \multicolumn{2}{c}{\quad\href{mailto:nsmattei@tulane.edu}{\small{\texttt{nsmattei@tulane.edu}}}} \\
        &&&\\
        \multicolumn{2}{c}{\textbf{Shivika Narang}}                                                                   & \multicolumn{2}{c}{\textbf{Sanjukta Roy}}\\
		\multicolumn{2}{c}{\small{UNSW Sydney, Australia}}                                                            & \multicolumn{2}{c}{\small{Indian Statistical Institute, India}}\\
		\multicolumn{2}{c}{\href{mailto:s.narang@unsw.edu.au}{\small{\texttt{s.narang@unsw.edu.au}}}\quad}            & \multicolumn{2}{c}{\quad\href{mailto:sanjukta@isical.ac.in}{\small{\texttt{sanjukta@isical.ac.in}}}} \\
	\end{tabular}
}

\date{}

\begin{document}
\begin{titlepage}

\maketitle

\begin{abstract}
\noindent We introduce a novel collective decision making problem that captures the ubiquitous issue of ordering food to cater for varied dietary preferences and requirements.  Our settings involve agents with  diverse dietary requirements over menu options with varied serving sizes. The goal is to select a menu where everyone has enough food they can consume and wastage of food is minimized. We introduce two different consumption models: optimistic and pessimistic. Optimistic consumption assumes a situation when a central planner can optimally allocate the food ordered among the agents to maximize the number of people who get enough to eat. Pessimistic considers the worst case guarantee on consumption when agents fill their own plates in an arbitrary order. Under either consumption model, we  seek valid menus (under which all agents are sufficiently fed) of minimum size.

Our work provides two sets of characterizations: (1) we characterize valid menus under either consumption model and (2) we characterize the space of instances that admit polynomial-time algorithms to find minimum sized menus. Our results also help us design Integer Linear Programs to find minimum sized menus in general settings. Furthermore, we present polynomial-time algorithms for important special cases. 
We then consider the worst case discrepancy between the size of minimum sized optimistic and pessimistic menus. We call this the {\em waste of pessimism} captured by the ratio of the minimum sized pessimistic menu to that of the minimum sized optimistic menu. We show tight upper bounds on this ratio. Our results also provide additional insights on the problem of finding a minimum sized maximal matching, which may be of independent interest. 

\end{abstract}

\end{titlepage}

\sloppy

\section{Introduction}

Deciding on a suitable menu catering to varied dietary requirements is a ubiquitous problem when sharing food 
--- a quintessential human activity that transcends borders and divisions. 
Everyday, millions of people debate on how to choose a non-wasteful menu order for a group to jointly consume at an eatery of some sort.  It is important to decide on an order which ensures enough food for everyone but at the same time is not so large that it leads to wastage.  
The choice of food to be shared may be affected by the various dietary restrictions and requirements among the diners. Dietary restrictions can stem from religious, ethical, health or many other considerations. For example, some people only eat vegan. Others avoid nuts because of allergies. Some cannot tolerate spice. 
In the broader sense, menu composition also comes into play when designing humanitarian food packages as outlined by United Nations High Commissioner for Refugees (UNHCR) with different standards for children, the elderly, infirm, and other groups.\footnote{\url{https://emergency.unhcr.org/emergency-assistance/food-security/emergency-food-assistance-standard}} When ordering food, we want to ensure that people have \textit{enough food} that they find feasible to eat but \textit{not so much that it leads to unnecessary waste}. The need for small sized menus may be from a variety of factors such as the food waste or the storage requirements from leftovers being minimized.  
We refer to the problem of finding a minimum sized menu that feeds everybody as the \emph{Menu Selection} problem.


For a given menu, what food is actually consumed, depends on the order in which people fill their plates and what they choose. If a central planner allocates the selected food among the agents, it is easier to ensure the consumption is maximized and wastage, in turn, is minimized. In contrast, if the agents fill their own plates coming in an {\em arbitrary order}, much more food may be wasted. A less restricted person may end up consuming food intended for someone who can only eat from a limited number of options. This may leave the person with more food restrictions with nothing to eat. When ordering food, one could potentially order a separate acceptable dish or menu option for each person. However, such an approach can be quite wasteful especially if the menu options serve multiple people. The central problem of this paper is how do we design a menu order that ensures that wastage is minimized and everyone is well-fed under pessimistic and/or optimistic perspectives. To this end, we introduce a formal conceptual framework which captures how the ordered menu is consumed. The smallest menu that feeds all agents will have the least food waste. Consequently, we focus on finding minimum sized menus that feed all agents. 

\begin{example}[Motivating Example]
	Consider an example where there are four agents $1, 2, 3, 4$ who want to  order something to eat together. The first three agents are vegetarian whereas agent 4 can also eat meat. 
    There are two food options $o_1$ (a vegetarian pizza) and $o_2$ (a chicken burger) with serving sizes 3 and 1 respectively.  Since agents $1, 2, 3$ are vegetarians, they can only eat the pizza and hence only find $o_1$ acceptable. On the other hand,  agent 4 finds both food options acceptable. 
    This acceptability relation can be represented as a bipartite graph (Figure~\ref{fig:graph}). 
    Note that since agents 1, 2, and 3 are vegetarian, they have a green leaf icon next to them and they only have an edge with option $o_1$ representing pizza (with a green pizza slice next to it). Agent 4 has edges with both food options. 
    
    The goal is to order sufficient food for the agents. If the order is $\{o_1,o_2\}$, it provides sufficient food and guarantees no food wastage for the agents as long as agent 4 does not eat any part of option $o_1$. There is risk that if agent 4 eats some of the pizza, then the other agents will not have a sufficient amount of acceptable food. 
On the other hand, if two units of $o_1$ are ordered, then the agents have sufficient food to eat irrespective of who eats what. In this case, 2 servings of food are wasted.
Ordering two units of $o_1$ is risk averse but it leads to food wastage. 
\label{example:basic}
\end{example}

\begin{figure}[t]
\centering
\scalebox{0.75}{
\begin{tikzpicture}[>=stealth]

    \tikzset{
        leaf/.pic={
            \fill[green!70!black] (0,0) .. controls (0.1,0.2) and (0.2,0.2) .. (0.2,0) 
                                  .. controls (0.2,-0.2) and (0.1,-0.2) .. (0,0) -- cycle;
        }
    }

    \def\yc{1}
    \def\xc{3}

    \filldraw[green!70!black] (\xc-2.4,\yc+1.0) node {\faLeaf};
    \node (1) at (\xc -2.0,\yc+1.0) {1};
 
    \filldraw[green!70!black] (\xc-1.1,\yc+1.0) node {\faLeaf};
    \node (2) at (\xc-0.7,\yc+1.0) {2};

    \filldraw[green!70!black] (\xc+0.3,\yc+1.0) node {\faLeaf};
    \node (3) at (\xc+0.7,\yc+1.0) {3};

    \node (4) at (\xc+1.8,\yc+1.0) {4};

    \filldraw[green!70!black] (\xc-2.4,\yc-1.0) node {\faPizzaSlice};
    \node (o1) at (\xc-2,\yc-1.0) {$o_1$};
    \node[below=0pt of o1] {pizza \quad $s_{o_1}=3$};

    \filldraw[orange!90!brown] (\xc+1.4,\yc-1.0) node {\faDrumstickBite};
    \node (o2) at (\xc+1.8,\yc-1.0) {$o_2$};
    \node[below=0pt of o2] {chicken \quad $s_{o_2}=1$};


    \draw[thick] (1) -- (o1);
    \draw[thick] (2) -- (o1);
    \draw[thick] (3) -- (o1);
    \draw[thick] (4) -- (o2);
    \draw[thick] (4) -- (o1);
    \draw (\xc,\yc+1.0) ellipse (3.2 and 0.5);
    \node at (\xc-5.5,\yc+1.0) {\textbf{Agents}};

    \draw (\xc,\yc-1.2) ellipse (4.0 and 1);
    \node at (\xc-5.5,\yc-1.0) {\textbf{Options}};

\end{tikzpicture}
}
\caption{A graphical representation of agents and options described in \cref{example:basic}, with vegetarian agents indicated by a leaf icon on the left.}
\label{fig:graph}
\end{figure}
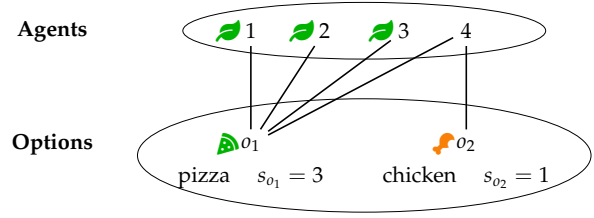
Ordering more food than necessary invariably leads to food waste.  Minimizing food waste at the consumer end is a concern across the developed world. Whether due to food being quick to spoil, limited to no storage for leftovers or the increased cost associated with ordering more food, when ordering food collectively, it is imperative to minimize food waste. On a global scale, halving the total amount of food wasted is one of the United Nations' Sustainable Development Goals, specifically target 12.3.\footnote{https://sdgs.un.org/goals/goal12} 
Some prior work has gone into algorithms for selecting menus \cite{fish2025stable} and on computationally minimizing food waste by learning food preferences and restrictions \cite{panda2019minimizing,flynn2025dish}. To the best of our knowledge, the objectives have not been studied together. We study a complementary setting where we are already equipped with the knowledge of individual food restrictions and aim to select a menu that can minimize food wastage, while feeding everyone sufficiently. 



\subsection*{Our Contributions}

\textbf{Conceptual Contribution.} We detail a new formal model for a natural collective choice problem that concerns ordering food \emph{with the goal of feeding a group while minimizing waste}. Our setting has a set of $n$ agents, a set of $m$ options each with associated serving size $s$, and each agent $i$ has an associated acceptability set $A_i$ consisting of all options $i$ can consume.  We develop two models of consumption to capture two extremes of how people may eat: \emph{optimistic and pessimistic}. 

Intuitively, these concepts are based on whether agents' plates are filled in a centralized or a decentralized manner. We call the former optimistic consumption, where a central entity is able to serve each agent food they find acceptable. We call the latter as pessimistic consumptions, we consider the worst case where agents may fill their plates themselves, in an arbitrary order. This setting allows for some agents to ``adversarially'' fill their plates in order to minimize the food available to a specific set of agents. For either consumption model, we wish to ensure that there is enough food for all agents. We call a menu which ensures that every agent is fed as optimistically or pessimistically {\em valid}. These terms are all defined in \cref{sec:model}. We defer an extended discussion on how our model relates to existing work to \cref{app:relwork}.

Further, this \textit{introduces a new line of questioning}: given a fixed algorithm and partial control on the input instance, can we ensure optimal outcomes with minimum inputs? Specifically, optimistic consumption is captured by a maximum cardinality matching between agents and the menu. In turn, pessimistic consumption is captured by the minimum sized maximal matching. Consequently, the problem of finding min sized valid menus is the same as ensuring the respective matching algorithm produce agent-saturating matchings.  In  \cref{sec:characterizations}, we show the equivalence between these matching based definitions and the motivating picking sequence approach. These characterizations help us design algorithms to find minimum sized menus. They also inform our subsequent results. 
Further, we adapt the well studied notion of price of anarchy to this setting and call it the {\em waste of pessimism}.  Specifically, this is the ratio of the sizes of the minimum sized pessimistically valid menu to that of the minimum sized optimistically valid menu. Our tools and algorithms could inform food delivery apps, humanitarian aid packaging, and even recommendations at restaurants in order to minimize waste and satisfy participating agents.   

Our results can be classified along the following broad directions, and our algorithmic and waste of pessimism results  are also summarized in \cref{tab:contributions}. 
\begin{enumerate}
     \item Characterizing valid menus under optimistic and pessimistic consumption, 
     \item Classifying instances that admit polynomial-time algorithms for computing minimum sized valid menus, and
     \item Obtaining tight bounds on the worst-case waste of pessimism.
\end{enumerate}

 \noindent\textbf{Computational Results.} We prove in \cref{subsec:complexity} that checking if a given instance admits an optimistically or pessimistically valid menu of size at most $B$ is NP-complete. For the optimistic case, this intractability result requires that serving sizes are at least $3$. Meanwhile, for the pessimistic case, the intractability holds even when all options have a serving size of one.  Based on the characterizations of valid menus, we design mixed integer linear programs (MILPs) for computing minimum sized valid menus in \cref{subsec:mips}.  While these MILPs gives us practical tools for computing minimum sized menus, they do not give good worst case guarantees. We develop a significantly faster algorithm, parameterized with respect to the number of agents $n$, for computing a minimum sized optimistically valid menu in \cref{subsec:fpt}. Note that we typically expect the number of options available $m$ exceed the number of agents. Consequently, \cref{alg:fpt} running in time $O^*(2^n)$ will be faster, in practice, than a naive algorithm which is exponential in the number of options.  

In \cref{sec:algorithms} we complement the intractability results by providing polynomial-time algorithms for computing minimum sized valid menus under several natural restrictions. 
We find that the efficient algorithms apply to several common realistic settings. For the case of optimistic consumption, we find polynomial-time algorithms for the following cases: (a) each option has a serving size of $1$ or $2$ (presented in \cref{subsec:smallplates}) and (b) Acceptability sets are {\em laminar:} for any two agents, their acceptability sets must be either disjoint or one is contained in the other (provided in \cref{subsec:laminar}). Laminar food restrictions seem quite natural: e.g., those who eat meat can eat vegetarian food as well, similarly, those who prefer vegetarian food, can have vegan food as well -- satisfying laminarity. For pessimistic consumption, we find a polynomial-time algorithm when acceptability sets are laminar. To this end, we refine the characterization for pessimistically valid menus for the laminar setting. These results help us in analyzing the waste of pessimism. 
\begin{table}[t]
	\centering
    \crefname{theorem}{Th.}{Th.}
    \crefname{proposition}{Prop.}{Prop.}
    \crefname{corollary}{Cor.}{Cor.}
	\scalebox{0.85}
    {
		\begin{tabular}{lclclcl}
			\toprule
            & \multicolumn{3}{c}{ Min Sized Valid Menus }& &\quad  {Waste of }\\
			& \quad Optimistic&  &   \quad Pessimistic\quad  & &\quad Pessimism \\
			\midrule
            \multirow{2}{*}{General} & \quad  NP-comp  & (\cref{thm:opt-menu-npc})  & \quad  NP-comp \quad         & (\cref{thm:pess-menu-npc})\quad              & \quad \multirow{3}{*}{$\geq \frac{n+4}{4}$}\\
                                     & $O^*(2^n)$      &(\cref{th:fpt})              \\
          $s_o\leq 2$               &in P &(\cref{thm:small-plates})                & \quad NP-comp  \quad  & (\cref{thm:pess-menu-npc})\quad  & \\
          Identical                 &in P & \multirow{3}{*}{\Bigg\}(\cref{thm:laminar-opt})}  & in P        &  \multirow{3}{*}{\Bigg\}(\cref{thm:laminar-pess})}  & $=1$   &\multirow{3}{*}{\Bigg\}(\cref{thm:wop-laminar})}\\         
          Chained                   &in P &                                         & in P                  &                                                     &$\leq 2$\\
          Laminar                   &in P &                                         & in P                  &                                                     & $\quad \leq  \frac{n+4}{4}$\\
			\bottomrule
		\end{tabular}
	}
	\caption{Summary of results on (1) computational complexity of  finding minimum sized menus and (2) known bounds on waste of pessimism.  }
	\label{tab:contributions}
\end{table}

\noindent  \textbf{Waste of Pessimism.} In \cref{sec:wop}, we study the worst case discrepancy between the size of minimum sized optimistic and pessimistic menus. We call this the {\em waste of pessimism}, captured by the ratio of the size of the minimum sized pessimistic menu to that of the minimum sized optimistic menu. We show tight upper bounds on this ratio for structured settings. 
Our main result is a tight bound for the case of laminar acceptability sets. Specifically, we show that given an instance with laminar acceptability relations, waste of pessimism is at most $\frac{n+4}{4}$, where $n$ is the number of agents. Interestingly, there is an instance where this bound is tight and the serving sizes are at most $2$. In contrast, we prove that the worst case bound to be $2$ when we assume {\em chained} acceptability sets, i.e., for any two agents, the acceptability set of one must be a subset of the other. Observe that the laminar setting is a generalization of the chained setting. This shows that the worst case blowup in menu size is not very significant when agents do not have restrictive requirements. 

\textbf{Consequences For Computing Minimum Sized Maximal Matchings.} Throughout the paper we use matchings as a tool to both describe consumption models as well as find solutions. Specifically, we use minimum sized maximal matchings to define pessimistic consumption. As a result, our characterization for pessimistically valid menus also provides insights on deciding given a bipartite graph, whether the minimum sized maximal matching saturates one side of the bipartition. Further, a consequence of our algorithm for laminar acceptability sets shows that a minimum sized maximal matching can be computed when neighborhoods of vertices are laminar.






Any omitted proofs are deferred to the appendix.



\section{Model}\label{sec:model}
We consider a set of agents $N=\{1,2,\dots, n\}$ and a set of options $O=\{o_1,\dots, o_m\}$. Each option can be viewed as a dish that could potentially feed multiple people. We assume all agents need the same amount of food and will not eat more than what they need. 
The term $s_o\in \mathbb{Z}_+\setminus\{0\}$ indicates the serving size of $o\in O$, that is,  how many people can be fed by one quantity of $o$.  For example, a pasta dish may serve one person, but a large pizza is intended to be shared by multiple people. 

Each agent $i\in N$ has a set of acceptable options $A_i\subseteq O$ that do not violate $i$'s personal dietary requirements. We assume that all options in $O\setminus A_i$ are not consumable by/unacceptable to $i$. Each agent has a requirement of getting a serving of one unit of food from any of the options that are acceptable to them. Note that agents do not express preferences over dishes beyond the dichotomous preferences of  acceptable vs unacceptable. Without loss of generality, we assume that for each $o\in O$, there exists some $i\in N$ s.t. $o\in A_i$; otherwise we could discard the option.

Our aim is to select a minimum sized menu that can feed all agents something they find acceptable, i.e., a menu order or, more simply, menu $M$ is a multiset of elements in $O$ along with the quantity selected for each option. We use the notation $q_o(M)$ to denote the number of occurrences of $o\in O$ within $M$. That is, $q_o(M)$ is the quantity of $o$ ordered under $M$. We use $size(M)=\sum_{o\in O}q_os_o$ to denote the size of menu $M$.
Given a set $S\subseteq O$, we shall use $M|_S$ to denote the menu $M'$ where $q_o(M')=q_o(M)$ for each $o\in S$ and $q_o(M')=0$ for each $o\in O\setminus S$.
We call our menu planning problem as the \Instance{} problem and denote an instance of it by $I=\langle N,O,(s_o)_{o\in O},(A_i)_{i\in N}\rangle$. 

{
\begin{example}[\Instance{} instance]
We now revisit the informal Example~\ref{example:basic} and reframe it in the notation introduced. 
Consider a problem with $4$ agents $N=\{1,2,3,4\}$ and two options $O=\{o_1,o_2\}$ where $s_{o_1}=3$ and $s_{o_2}=1$, and $A_1=A_2=A_3=\{o_1\}$ and $A_4=\{o_1,o_2\}$. 
    \end{example}
}


\subsection*{Consumption Models}

When ordering food, we typically have some natural constraints. Our primary constraint is that the selected menu should ensure that {\em each agent has enough food to eat}. Further, it is important to note that if agents are allowed to fill their plates, however they wish, one at a time in some order, then a menu that might suffice for everyone under one ordering of the agents may not be sufficient under another. To this end, we study two consumption models: optimistic and pessimistic.
We say that a menu is valid under a specific model, if under that model, there is enough acceptable food for each agent. In order to capture the ways in which agents may consume from a menu, we introduce the following tool:

\begin{definition}[Consumption Graph]
    Given a \Instance{} instance $I=\langle N,O,(A_i)_{i\in N}$, $(s_o)_{o\in O}\rangle$  and a menu $M$, the consumption graph of $M$, denoted by $G_M(I)=(N,O_M,E_M)$, is a bipartite graph where there is a vertex for each agent $i$ and $q_o(M)s_o$ vertices for each option $o\in O$. There is an edge from $i$ to a copy of $o$ s.t. $q_o(M)\geq 1$ if and only if $o\in A_i$.
\end{definition}

We shall capture how agents fill their plates from the selected menu by a matching on this consumption graph. Given a menu $M$, $x$ is a \emph{feasible matching} if $x:N\times O\rightarrow [0,1]$, s.t. $\sum_{o\in A_i}x_{i,o}\leq 1$,  $\sum_{i\in N}x_{i,o}\leq q_o(M)s_o$ and $x_{i,o}>0$ only if $o\in A_i$. That is, $x_{i,o}$ denotes the amount of option $o$ that $i$ consumes, where $i$ only consumes acceptable options and the maximum amount of $o$ consumable is $q_o(M)s_o$, where $q_o(M)$ is the quantity of $o$ ordered. Note that we use a definition of matchings consistent with linear programs, in order to simplify our discussion of ILPs later. 
Furthermore, it is well known that the linear polytope is integral \cite{plummer1986matching}. Consequently, we shall assume, without loss of generality, that any feasible matching $x$ that we compute is integral.

\subsection*{Optimistic Consumption} Under the optimistic consumption model, we effectively assume that a central planner fills everyone's plates such that all agents have enough acceptable food to eat. This can be seen as a setting where agents' plates are being prepared beforehand and being served to the agent at the table. 

\begin{definition}[Optimistically Consumed Amount $\lambda$]
    Given a \Instance{} instance $I=\langle N,O,(A_i)_{i\in N}$, $(s_o)_{o\in O}\rangle$ and a menu $M$, let $x$ be the maximum sized matching of the consumption graph $G_M(I)$. The optimistically consumed amount of $M$ is $\lambda(M)=\sum_{(i,o)\in N \times O}x_{i,o}$, that is, the size of $x$.
\end{definition}

{Optimistic Consumption setting can be thought of as a generalization of set cover. In set cover, each set covers all its constituent elements but in our setup, one option can only cover the corresponding serving size. } As we need to feed all agents, we look for optimistically valid menus.

\begin{definition}[Optimistically Valid Menus]
    Given a \Instance{} instance $I=\langle N,O,(A_i)_{i\in N}, $$(s_o)_{o\in O}\rangle$, a menu $M$ is valid under optimistic consumption, or is {\em optimistically valid}, if $\lambda(M)=n$. We use $\Lambda(I)$ to denote the space of {\em all} optimistically valid menus. 
\end{definition}
That is, $M$ is optimistically valid if there exists a feasible matching $x:N\times O\rightarrow [0,1]$, s.t. $\sum_{o\in O} x_{i,o}=1$.
Observe that if there is such a matching, we can divide the food among the $n$ agents that each agent gets a sufficient  amount of something they find acceptable. 
 We shall call this the {\em optimistic consumption matching.}

Observe that optimistic consumption is not always feasible in practice, for instance at children's birthday parties or in airplanes. Consider when there are two passengers on an airplane and one of them is vegan. Optimistically, it would suffice to have only one vegan meal and one meat based meal. If the non-vegan passenger is asked for their meal choice first and they take the vegan meal, then the vegan passenger is left with nothing to eat. To this end, we also study pessimistic consumption. 

\subsection*{Pessimistic Consumption} Under pessimistic consumption, we allow agents to fill their plate in any manner, so long as they only eat from acceptable options and do not purposefully waste food. That is, we only consider those settings, where if an agent has a plate that is not full, there should be nothing left for them to eat. Allowing for complete autonomy in this manner, we wish to ensure that there should be enough food for each agent. We capture such consumption using \textit{maximal matchings}. 

Formally, we call a matching $x$ for the consumption graph $G_M$ \textit{maximal} if for each $i\in N$ s.t. $\sum_{o\in O}x_{i,o}<1$, we have that for each option $o\in \{o\in A_i|q_o(M)>0\}$, it holds $\sum_{i'\in N}x_{i',o}=q_o(M)s_o$.

\begin{definition}[Pessimistically Consumed Amount $\gamma$]
    Given a \Instance{} instance $I=\langle N,O,(A_i)_{i\in N}$, $(s_o)_{o\in O}\rangle$ and a menu $M$, let $x$ be the minimum sized maximal matching of the consumption graph $G_M(I)$. The pessimistically consumed amount of $M$ is $\gamma(M)=\sum_{(i,o)\in N \times O}x_{i,o}$, that is, the size of $x$.
\end{definition}

\begin{definition}[Pessimistically Valid Menus]
    Given a \Instance{} instance $I=\langle N,O,(A_i)_{i\in N}, (s_o)_{o\in O}\rangle$, a menu $M$ is valid under pessimistic consumption, or is {\em pessimistically valid}, if $\gamma(M)=n$. We use $\Gamma(I)$ to denote the space of {\em all} optimistically valid menus. 
\end{definition}
Analogous to the optimistic setting, we shall call the minimum sized maximal matching $x:N\times O\rightarrow [0,1]$ as the {\em pessimistic consumption matching}. 


{
\begin{example}[Optimistically and Pessimistically Valid Menus]\label{ex:menus}
Consider Example~\ref{example:basic} with $4$ agents $N=\{1,2,3,4\}$ and two options $O=\{o_1,o_2\}$ where $s_{o_1}=3$ and $s_{o_2}=1$, and   $A_1=A_2=A_3=\{o_1\}$ and $A_4=\{o_1,o_2\}$. 
Then $\{o_1,o_2\}$ is an optimistically value menu. However, it is not pessimistically valid. The menu $\{o_1,o_1\}$ is pessimistically valid. We illustrate these menus in \cref{fig:consumption-graph}.
    \end{example}
}

\begin{figure*}[t]
\centering

\scalebox{0.7}{
\begin{tikzpicture}[>=stealth]

    \def\yc{1}
    \def\xc{5}

    \filldraw[green!70!black] (-0.3,\yc+1.0) node {\faLeaf};
    \node (i1) at (0,\yc+1) {1};

    \filldraw[green!70!black] (-0.3,\yc+0.35) node {\faLeaf};
    \node (i2) at (0,\yc+0.35) {2};

    \filldraw[green!70!black] (-0.3,\yc-0.35) node {\faLeaf};
    \node (i3) at (0,\yc-0.35) {3};

    \node (i4) at (0,\yc-1) {4};

    \node (o1) at (\xc-2,\yc+0.5) {$o_1$};

    \node (o2) at (\xc-2,\yc-0.5) {$o_2$};

    \filldraw[thick,red!80!black] (i1) -- (o1);
    \filldraw[thick,red!80!black] (i2) -- (o1);
    \filldraw[thick,red!80!black] (i3) -- (o1);
    \filldraw[thick,red!80!black] (i4) -- (o2);
    \draw[dash pattern={on 3pt off 1pt}] (i4) -- (o1);
    \draw[thick] (0,\yc) ellipse (1.0 and 1.4);
    \node at (0,\yc+1.7) {\textbf{Agents}};

    \draw[thick] (\xc-2,\yc) ellipse (0.6 and 1.2);
    \node at (\xc-2,\yc+1.7) {\textbf{Optimistic Menu}};
    \node at (1.7,\yc-1.7) {$M_1=\{o_1,o_2\}$};

    \filldraw[green!70!black] (\xc+1.7,\yc+1.0) node {\faLeaf};
    \node (j1) at (\xc+2,\yc+1) {1};

    \filldraw[green!70!black] (\xc+1.7,\yc+0.35) node {\faLeaf};
    \node (j2) at (\xc+2,\yc+0.35) {2};

    \filldraw[green!70!black] (\xc+1.7,\yc-0.35) node {\faLeaf};
    \node (j3) at (\xc+2,\yc-0.35) {3};

    \node (j4) at (\xc+2,\yc-1) {4};

    \node (p1) at (\xc+5,\yc+0.5) {$o_1$};

    \node (p2) at (\xc+5,\yc-0.5) {$o_1$};

    \filldraw[thick,red!80!black] (j1) -- (p1);
    \filldraw[thick,red!80!black] (j2) -- (p1);
    \draw[dash pattern={on 3pt off 1pt}] (j3) -- (p1);
    \draw[dash pattern={on 3pt off 1pt}] (j4) -- (p1);
    \draw[dash pattern={on 3pt off 1pt}] (j1) -- (p2);
    \draw[dash pattern={on 3pt off 1pt}] (j2) -- (p2);
    \filldraw[thick,red!80!black] (j3) -- (p2);
    \filldraw[thick,red!80!black] (j4) -- (p2);
    \draw[thick] (\xc+2,\yc) ellipse (1.0 and 1.4);
    \node at (\xc+2,\yc+1.7) {\textbf{Agents}};

    \draw[thick] (\xc+5,\yc) ellipse (0.6 and 1.2);
    \node at (\xc+5,\yc+1.7) {\textbf{Pessimistic Menu}};
    \node at (\xc+3.7,\yc-1.7) {$M_2=\{o_1,o_1\}$};
\end{tikzpicture}
}
\caption{Consumption graphs of the minimum sized menus in \cref{ex:menus}. The corresponding optimistic and pessimistic matchings are drawn in {\color{red!80!black} red}. Edges not in the matching are dashed.}
\label{fig:consumption-graph}
\end{figure*}
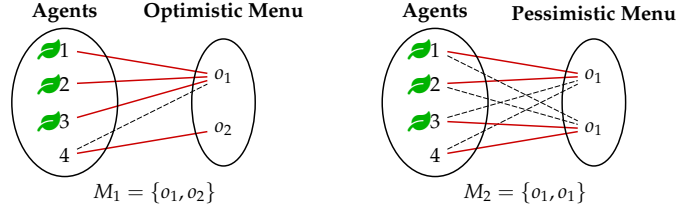

\begin{remark}
     Consider an instance where for each agent $|A_i|= 1$. In this case, we need to order the one item this agent finds acceptable or they will not be able to eat. Specifically, for each option $o\in O$, we need to order $\left\lceil \frac{|\{i\in N|o\in A_i\}|}{s_o} \right\rceil$ quantities of $o$. Recall that for each option, we assume at least one agent finds it acceptable. In the case where $|A_i|=1$, each option must be ordered. This is a minimum sized optimistically and pessimistically valid menu. 
\end{remark}

\section{Characterizations of Valid Menus}\label{sec:characterizations} 
In this section, we provide characterizations of optimistically and pessimistically valid menus. These characterizations also help connect our graph theoretic definitions of consumption to the motivation: ensure that each agent is fed one unit of acceptable food.  To this end, we first define a consumption sequence to denote the order in which the options in a given menu are consumed.  
Recall that given $S\subset O$ and a menu $M$, $M|_S$ denotes the multiset that contains as many copies of $o\in S$ as $M$ and no copies of options not in $S$. For the purposes of this section and its proofs, we shall use the notation $\widehat{M}$ to denote the multiset containing $q_o(M)s_o$ quantities of each option $o\in O$. We say that $\pi$ is an ordering over $N$ or a {\em picking sequence} if $\pi:[n]\rightarrow N$ and $\pi$ is a bijection. 

\begin{definition}[Possible Outcome]
    Given a \Instance{} instance $I$ and a menu $M$ and a picking sequence $\pi$, we say that $\omega:[n]\rightarrow O\cup \emptyset$ is a possible outcome of $\pi$ on $M$ if  
    for each $k\in [n]$, we have that $\omega(k)\in \widehat{M}^k|_{A_{\pi(k)}}$ where $\widehat{M}^k=\widehat{M}\setminus \left(\bigcup_{k'<k} \omega(k')\right)$, i.e., $\widehat{M}^k$ is the set of available options after the first $k-1$ agents have picked and $\widehat{M}^k|_{A_{\pi(k)}}$ is the set of options in $\widehat{M}^k$ that are acceptable by the agent $\pi(k)$. Further, it must hold that $\omega(k)=\emptyset \Leftrightarrow \widehat{M}^k|_{A_{\pi(k)}}=\emptyset$. We use $\Omega(\pi,M)$ to denote the space of all possible outcomes of $\pi$.
    \end{definition}

In other words,  $\omega$ captures one possible outcome of agents picking food according to picking sequence $\pi$ s.t. at the $k^{\text{th}}$ turn agent $\pi(k)$ picks one unit of an acceptable option from $M$, if any are left. Recall that our motivation for the optimistic case was to ensure that there was some way to ensure that each agent finds something acceptable to eat. The following characterization shows that our definition does precisely this. 

\begin{restatable}{theorem}{optcharac}[Characterization of optimistically valid menus]
Given  a \Instance{} instance $I=\langle N,O,(A_i)_{i\in N}, (s_o)_{o\in O}\rangle$ and a menu $M$, the following conditions are equivalent: 
\begin{enumerate}
    \item $M$ is optimistically valid;
    \item The consumption graph $G_M$ admits matching of size $n$;
    \item There exists a picking sequence $\pi$ and a possible outcome $\omega$ of $\pi$ on $M$ s.t. for each $k\in [n]$, $\omega(k)\neq \emptyset$.  
\end{enumerate}
\label{th:charact-opt}
\end{restatable}

\begin{remark}
    Note that in the proof of the above theorem, when constructing a consumption sequence for an optimistic menu, the exact ordering of agents is irrelevant, as long as their chosen option is consistent with a maximum cardinality matching of the consumption graph $G_M$.
\end{remark}

We now turn to pessimistic menus. Recall that a menu $M$ is pessimistically valid if {\em every} minimum sized maximal matching of $G_M$ has size $n$. Note that,  the exact  computation of a minimum sized maximal matching is not essential to check for pessimistic validity. It's only necessary to check its size is $n$. 
Specifically, even those maximal matchings that capture adversarial behavior towards any one agent need to have size $n$.  To this end, we need that for each agent $i\in N$, there should be more food ordered from $A_i$ than the other agents $N\setminus \{i\}$ combined can eat. We capture this risk to pessimistic validity as follows:

\begin{definition}[Pessimistic Risk $\rho$]
     Given a \Instance{} instance $I$ and a menu $M$, for an agent $i\in N$, the pessimistic risk under $M$ denoted by $\rho_i(M)$ is the size of the maximum cardinality matching between $N\setminus\{i\}$ and $M|_{A_i}$.
\end{definition}

We find that this definition helps us succinctly capture a necessary and sufficient condition for pessimistically valid menus. 

\begin{restatable}{theorem}{pesscharac}[Characterization of pessimistically valid menus]\label{lem:pess-valid}
Given  a \Instance{} instance $I=\langle N,O,(A_i)_{i\in N}, (s_o)_{o\in O}\rangle$ and a menu $M$, the following conditions are equivalent:
\begin{enumerate}
    \item $M$ is pessimistically valid; 
    \item For every agent $i\in N$, we have that $size(M|_{A_i})=\sum_{o\in A_i} q_o(M)s_o\geq 1+ \rho_i(M)$;
    \item For every picking sequence $\pi$ and every $\omega\in \Omega(\pi,M)$, we have that for each $k\in [n]$, $\omega (k)\neq \emptyset$.  
\end{enumerate}
\end{restatable}

\begin{proof}[Partial Proof]
    We defer the proof for the equivalence of (1) and (2) to \cref{app:charac}. Assuming this equivalence holds, we now prove the remaining case.

    \textbf{Equivalence of statement $(3)$ to pessimistic validity.} Let $M$ be pessimistically valid. Consider a consumption sequence $\omega$ of $M$ with a corresponding ordering $\pi$. Recall that for any $k\in [n]$, each preceding agent, under $\pi$, can consume at most one unit of an option acceptable to agent $i$.  Consequently, 
    we have that the difference between the number of units of acceptable options for agent $\pi(k)$ under $M$ and that under $\widehat{M}^k$ is $size(\widehat{M}|_{A_{\pi(k)}})-size(\widehat{M}^k|_{A_{\pi(k)}})\leq \rho_i$. As $M$ is pessimistically valid, from (2) we have that $size(M|_{A_i})=\sum_{o\in A_i} q_o(M)s_o\geq 1+ \rho_i(M)$. Consequently, $\widehat{M}^k|_{A_{\pi(k)}}$ must be non-empty. As a result, $\omega(k)\neq \emptyset$ for each $k\in [n]$.

    Finally, for the converse, let $M$ be a menu s.t. for every consumption sequence $\omega$ of $M$, we have that for each $k\in [n]$, $\omega (k)\neq \emptyset$. Fix agent $i\in N$. Let $x$ be maximum cardinality matching between $N\setminus \{i\}$ and $M|_{A_i}$ . Consider an arbitrary ordering $\pi$ under which $\pi(n)=i$. That is, $i$ is the last agent to pick from $M$. Let $N^i=\{j\in N\setminus \{i\}|\sum_{o\in A_j} x_{j,o}=1\}$.  Consider a choice of $\omega$ s.t. for $\pi(k)\in N^i$  and the option $o$ for which $x_{\pi(k),o}=1$, we set $\omega(k)=o$. For all $k\in [n]$ s.t. $k\notin N^i$, choose $\omega (k)$ arbitrarily from $\widehat{M}^k$, whenever $\widehat{M}^k$ is non-empty.     
    It is straightforward to see that $\omega$ is indeed a consumption sequence for $M$. Consequently, it must hold that $\omega(n)\neq \emptyset$, by assumption. As a result, we must have that $\widehat{M}^n|_{A_i}$ must be non-empty. This gives us that the size of $x$ if $\rho_i=\sum_{j\neq i}\sum_{o\in A_j}x_{j,o}<\sum_{o\in A_i} q_o(M)s_o$. 
    
    Hence, we have that for each $i\in N$, $size(M|_{A_i})>\rho_i$ and thus, $M$ is pessimistically valid.     
\end{proof}

\section{Computing Minimum Sized Valid Menus}\label{sec:complexity-ILPs}

We now explore the computational aspects of finding minimum sized menus that are optimistically or pessimistically valid. As mentioned before, this problem is a generalization of the $\textsc{SetCover}$ problem so it is unsurprising that it proves to be computationally intractable. We first discuss the complexity of finding minimum sized menus, then give exact algorithms for finding minimum sized optimistically valid menus. 

\subsection{Computational Complexity}\label{subsec:complexity}
Before settling the complexity of computing a minimum sized menu, we first discuss the complexity of verification of validity. Observe that one way of checking if a menu is optimistically or pessimistically valid is to compute the optimistically and pessimistically consumed amount.

\begin{restatable}{proposition}{validverify}\label{thm:verify}
    Given a \Instance{} instance $I$ and a menu $M$ 
    \begin{enumerate}
        \item  the  amount of $M$ that is  optimistically consumed $\lambda(M)$ can be computed in polynomial time, 
        \item  it is NP-hard to compute the amount of $M$ that is pessimistically consumed $\gamma(M)$, and
        \item  it can be checked in polynomial time if $M$ is optimistically and/or pessimistically valid. 
    \end{enumerate}   
\end{restatable}


\begin{remark}
    Given a graph $G=(V,E)$ using \Cref{thm:verify}, we can check in polynomial time if the minimum sized maximal matching has size $|V|$.
\end{remark}

We now show that finding a minimum sized optimistically valid menu is NP-hard even when each agent finds at most two options acceptable and each option has serving size three. The proof is via a reduction from the \textsc{3DimensionalMatching} (3DM) problem. 

\begin{sloppypar}
\begin{restatable}{proposition}{optcomplexity}\label{thm:opt-menu-npc}
    Given a \Instance{} instance $I=\langle N,O, (A_i)_{i\in N},(s_o)_{o\in O}\rangle$ and size $B\geq n$ where $s_o\geq 3$, it is NP-complete to decide if there exists an optimistically valid menu of size at most $B$.
\end{restatable} 
\end{sloppypar}

We show the NP-hardness of finding minimum sized pessimistically valid menus via a reduction from the Exact Cover (X3C) problem, which is known to be NP-complete. \cite{garey1979computers}. 

\begin{restatable}{proposition}{pesscomplexity}\label{thm:pess-menu-npc}
    Given a \Instance{} instance $I=\langle N,O, (A_i)_{i\in N},(s_o)_{o\in O}\rangle$ and size $B\geq n$ where $s_o = 1$, it is NP-complete to decide if there exists a pessimistically valid menu of size at most $B$.
\end{restatable}

\begin{proof}[Proof Sketch]
We give a reduction from the \textsc{Exact3Cover} (X3C) problem, which is known to be NP-hard \cite{garey1979computers}. It is formally defined as follows: 

\begin{center}
\begin{tabularx}{1.0\columnwidth}{l}
\toprule
\multicolumn{1}{c}{\textbf{\textsc{Exact3Cover}} (X3C)} \\
\midrule
\parbox[t]{0.95\columnwidth}{\textbf{Given:}  A universe $\mathcal{U}$ of $3k$ elements and a family of $\ell$ subsets $\mathcal{S}$, where for each $S_t\in \mathcal{S}$, we have that $|S_t|=3$.}\\
[1em]

\parbox[t]{0.95\columnwidth}{\textbf{Question:} Do there exist $k$ mutually disjoint subsets $\{S_{i_1},\cdots,S_{i_k}\}$ which cover $\mathcal{U}$, that is, $\cup_{j\in [k]}S_{i_j}=\mathcal{U}$?}\\
\bottomrule
\end{tabularx}
\end{center}
     Given an instance of X3C $\langle \mathcal{U},\mathcal{S}\rangle$, we create a \Instance{} instance as follows: For each element $e\in \mathcal{U}$, create element agent $i_e$. For each subset $S_t\in \mathcal{S}$, create three set agents $j^t_1$, $j^t_2$ and $j^t_3$ and options $o^t_1,\cdots, o^t_4$. For each $o\in O$, we set $s_o=1$. For each $S_t\in \mathcal{S}$, let $S_t=\{e^t_1,e^t_2,e^t_3\}$. For each $z\in [3]$, $A_{j^t_{z}}=\{o^t_{z},o^t_4\}$. Furthermore, add $o^t_{z}$ to $A_{i_{e_{z}}}$ for each $z\in [3]$.   
   
   Consequently, we create an instance with $n=3(k+\ell)$ agents and $m=4\ell$ options. For each $e\in \mathcal{U}$, $A_{i_e}$ contains exactly one option corresponding to each subset containing $e$. We illustrate a set gadget for set $S_t=\{e_1,e_2,e_3\}$ in \cref{fig:gadget}. We show than an X3C exists if and only if there is a pessimistically valid menu of size $n=3(k+\ell)$. Note that if an X3C exists, say $\{S_{i_1},\cdots,S_{i_k}\}$, we can create a pessimistically valid menu $M$ of size $n$. For each set $S_t\in \{S_{i_1},\cdots, S_{i_k}$, we set $q_{o^t_1}(M)=q_{o^t_2}(M)=q_{o^t_3}(M)=2$ and $q_{o^t_4}(M)=0$. For each set $S_t\notin \{S_{i_1},\cdots,S_{i_k}\}$, we set $q_{o^t_1}(M)=q_{o^t_2}(M)=q_{o^t_3}(M)=0$ and $q_{o^t_4}(M)=3$.
\end{proof}

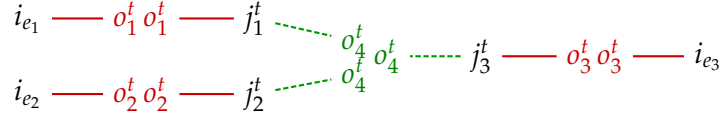
\begin{figure}[t]
\centering
\scalebox{1}{
\begin{tikzpicture}[>=stealth]

    \def\yc{0.5}
    \def\xc{6}

    \node (e1) at (\xc-4.5,\yc+0.5) {$i_{e_1}$};
 
    \node (oe1) at (\xc-3.2,\yc+0.5) {\color{red!80!black}$o^t_1$};

    \node (oj1) at (\xc-2.8,\yc+0.5)  {\color{red!80!black}$o^t_1$};

    \node (j1) at (\xc-1.5,\yc+0.5) {$j^t_1$};

    \node (o14) at (\xc-0.18,\yc+0.2)  {\color{green!50!black} $o^t_4$};
    \node (o24) at (\xc-0.18,\yc-0.25)  {\color{green!50!black} $o^t_4$};
    \node (o34) at (\xc+0.25,\yc)           {\color{green!50!black} $o^t_4$};

    \node (e2) at (\xc-4.5,\yc-0.5) {$i_{e_2}$};
 
    \node (oe2) at (\xc-3.2,\yc-0.5) {\color{red!80!black}$o^t_2$};
    \node (oj2) at (\xc-2.8,\yc-0.5) {\color{red!80!black}$o^t_2$};

    \node (j2) at (\xc-1.5,\yc-0.5) {$j^t_2$};

    \node (e3) at (\xc+4.5,\yc)     {$i_{e_3}$};
 
    \node (oe3) at (\xc+3.2,\yc)     {\color{red!80!black}$o^t_3$};
    \node (oj3) at (\xc+2.8,\yc) {\color{red!80!black}$o^t_3$};

    \node (j3) at (\xc+1.5,\yc)     {$j^t_3$};

    \filldraw[thick,red!80!black] (e1) -- (oe1);
    \filldraw[thick,red!80!black] (oj1) -- (j1);
    \filldraw[thick,green!60!black][dash pattern={on 2pt off 1pt}] (j1) -- (o14);
    \filldraw[thick,red!80!black] (e2) -- (oe2);
    \filldraw[thick,red!80!black] (oj2) -- (j2);
    \filldraw[thick,green!60!black][dash pattern={on 2pt off 1pt}] (j2) -- (o24);
    \filldraw[thick,red!80!black] (oe3) -- (e3);
    \filldraw[thick,red!80!black] (oj3) -- (j3);
    \filldraw[thick,green!60!black][dash pattern={on 2pt off 1pt}] (j3) -- (o34);

\end{tikzpicture}
}
\caption{Consumption graph for set gadget for a set $S_t=\{e_1,e_2,e_3\}$ with the case for $S_t$ in X3C shown in {\color{red!80!black} solid red} edges and not in $S_t$ shown in {\color{green!50!black} dashed green} edges.}
\label{fig:gadget}
\end{figure}

\subsection{Mixed Integer Linear Programs for Minimum Sized Valid Menus}\label{subsec:mips}
We now provide Mixed Integer Linear Program (MILP) formulations for finding minimum sized valid menus. MILP solvers allow us to find solutions quickly, in practice. Throughout this discussion, we shall slightly overload our notation and use $q_o$ to denote the chosen amount of option $o$. Observe that every choice of $(q_o)_{o\in O}$ has a bijection to the space of menus. Consequently, we shall talk about a choice of $(q_o)_{o\in O}$ and its corresponding menu interchangeably. Further, variables $x_o$, $x^i_o$ and $y_j$ shall be binary variables denoting whether the corresponding option/agent is selected. 

\subsubsection*{Optimistic Consumption.} We now present an MILP formulation for computing a minimum sized optimistic menu. Here, we take a \Instance{} instance $I=\langle N,O, (A_i)_{i\in N}, (s_o)_{o\in O}\rangle$ as input. We use $m$ variables $q_o$ for each $o\in O$ to denote the quantities of each option $o$. We further denote $x_{o}^i$ to denote a maximum cardinality matching between agents and their respective acceptable options. The MILP below is essentially adding minimum sized quantity objective function on top of a classical bipartite matching linear program. 

\begin{equation*}
\begin{array}{ll@{}lll}
\text{minimize}   & \displaystyle\sum\limits_{o\in O}   &q_{o}s_{o}  &\\
\text{subject to} & \displaystyle\sum\limits_{i\in N}   &x_{o}^i     &\leq q_os_o,           &o\in O\\
                  & \displaystyle\sum\limits_{o\in A_i} &x_{o}^i     &\geq 1,                &i\in N\\
                  &                                     &x_{o}^i     &\in [0,1],             &i\in N, o\in O\\
                  &                                     &q_o         &\in \{0,1,\cdots,\lceil \frac{n}{s_o}\rceil\}~~~&o\in O
\end{array}
\end{equation*}

The MILP above has $m+nm$ variables and $2m+2n+mn$ constraints. 
Observe that for a menu that is not optimistically valid, by definition, no matching can ensure that $ \sum_{o\in A)i} x_{o}^i\geq 1$. Consequently, the polytope defined in the MILP above is feasible for a choice of $(q_o)_{o\in O}$ if and only if the corresponding menu is  optimistically valid. Further, as optimistically valid menus are guaranteed to exist, there is always a solution for which the MILP is feasible. As a result, given a \Instance{} instance $I$, the MILP finds a minimum sized optimistically valid menu.


\subsubsection*{Pessimistic Consumption.} We now look for an MILP for pessimistic consumption. We first give an ILP to check if a given menu is pessimistically valid. Consider \cref{lem:pess-valid}, where, to check for pessimistic validity for each agent $i\in N$, we compute a maximum cardinality matching from $N\setminus\{i\}$ to options in $A_i$. Finding such a matching with an LP can be done by setting the objective function to maximize the size of the matching. However, a construction which puts the matching size in the objective is not useful when building minimum sized pessimistically valid menus. Consequently, we require a polytope which is feasible if and only if the menu is pessimistically valid. Unlike the case for optimistic consumption, we do not want to be able to match all agents in $N\setminus\{i\}$. Ideally, we would like to find menus that minimize the size of this maximum cardinality matching. 

There {seems to be} no straightforward way to encode a maximum cardinality matching, where all agents need not be matched, solely through linear constraints. However, recall that we do not actually need to compute a maximum cardinality matching. From \cref{lem:pess-valid}, we only require that each maximum cardinality matching from $N\setminus \{i\}$ to $M|_{A_i}$ should leave some acceptable food for agent $i$. In other words, we have that $M$ is pessimistically valid for $i$ if and only if there is no option saturating matching for the graph between $N\setminus \{i\}$ and $M|_{A_i}$. From Hall's theorem \cite{hall1935representatives}, we can certify this without computing a maximum cardinality matching. Specifically, there must be some set of options in $M|_{A_i}$ for which the condition in Hall's theorem is violated. That is, there must exist $S\subseteq A_i$ s.t. $\sum_{o\in S} q_os_o > |\{j\in N\setminus\{i\}:A_j\cap S\neq \emptyset\}|$.  We capture this in the following ILP which will be feasible if and only if such a set $S$ exists. In this way, we combine our characterization in \cref{lem:pess-valid} with Hall's theorem to checkpessimistic validity through an ILP.

\paragraph{ILP Setup} Given a menu $M$, we use the following ILP to check pessimistic validity for a fixed $i\in N$. Let $N^i=\{j\in N\setminus\{i\}|\exists o\in A_i\cap A_j\}$. Observe that only the agents in $N^i$ can be matched in a maximum cardinality matching from $N\setminus \{i\}$ to $M|_{A_i}$. We use $x_o$ to denote where $o$ is in $S$ and we set $y_j$ to $1$ iff for some $o\in A_j\cap A_i$ we have that $x_o=1.$  We use variables of type $x_o$ to select options and $y_j$ to select agents in $N\setminus \{i\}$. Consequently, we always choose $x_o,y_j\in \{0,1\}$. 
Thus, the amount of food within $S$ is captured by $\sum_{o\in A_i} x_oq_os_o$, where $x_o=1 \Leftrightarrow o\in S$. For $S$ to violate Hall's theorem, we have that the size of its neighborhood set, which is captured by $\sum_{j\in N\setminus \{i\}}y_j<\sum_{o\in A_i} x_oq_os_o$.

\begin{align}
\text{minimize ~~~~~~ }   &1\nonumber\\
\text{subject to ~~~~~~}    \nonumber\\
        \displaystyle\sum\limits_{o\in A_i} x_oq_os_o   &  \geq 1+ \displaystyle\sum\limits_{j\in N^i}y_j           && \text{(Hall's theorem violation certificate)}  \\ 
        y_j                                             &  \geq x_o                                                 && j\in N^i, o\in A_i\cap A_j\neq\emptyset \label{eq:yj-lower}\\
        y_j                                             &  \leq\displaystyle\sum\limits_{o\in A_i\cap A_j} x_o      && j\in N^i \label{eq:yj-upper}\\ 
        x_o                                             &  \in \{0,1\}                                              && o\in A_i\nonumber\\ 
        y_j                                             &  \in \{0,1\}                                              && j\in N^i\nonumber 
\end{align}

Observe that the ILP above has at most $n+m$ variables and at most $nm+2n+m+1$ constraints. 

\paragraph{Correctness of ILP} We now show that this ILP is feasible if and only if the given menu $M$ is pessimistically valid for $i$. 
Let the above ILP be feasible. Let $((x_o)_{o\in A_i},(y_j)_{j\in N^i})$ be a feasible solution of the ILP. Choose the set of options $S=\{o\in A_i|x_o=1\}$. Observe that this must be non-empty under any feasible solution. The amount of capacity across these options is $\sum_{o\in S}q_os_o=\sum_{o\in A_i}x_oq_os_o$. 

Observe that for each $j\in N^i$ s.t. there exists $o\in S\cap A_j$. For each of these agents, from constraint \ref{eq:yj-lower}, we need $y_j\geq x_o=1$. Consequently, it must hold that $y_j=1$, as each $y_j\in \{0,1\}$. Further, for all agents $j\in N\setminus \{i\}$ s.t. $A_j\cap S=\emptyset$, by definition of $S$, we have that for each option $o\in A_j\cap A_i$, it must hold that $x_o=0$. Thus, by constraint \ref{eq:yj-upper}, it must hold that $y_j=0$. Consequently, we have that $\sum_{j\in N^i}y_j=|\{j\in N^i|\exists o\in S\cap A_j\}|$ which is precisely the size of neighborhood of $S$.  

As $((x_o)_{o\in A_i},(y_j)_{j\in N^i})$ is a feasible solution of the ILP, we have that $\sum_{o\in A_i}x_oq_os_o>\sum_{j\in N^i}y_j$. Thus, by Hall's theorem,  no matching in the consumption graph of $M|_{A_i}$ can saturate the options in $M|_{A_i}$. Consequently, $M$ must be pessimistically valid for $i$.

For the converse, let $M$ be pessimistically valid for $i$. Thus, there must exist a set of options $S\subseteq A_i$ s.t. $\sum_{o\in S} q_os_o>|\{j\in N^i|\exists o\in S\cap A_j\}|$. For each option $o\in A_i$, set $x_o=1\Leftrightarrow o\in S$. For each $j\in N^i$, set $y_j=\min \{1,\sum_{o\in A_j\cap A_i}x_o\}$. It is straightforward to verify that this is a feasible solution for the ILP above. 

\textbf{Constructing Pessimistic Menus.} 
We now extend the ideas in the previous ILP to get a minimum sized pessimistically valid menu for a given instance $M$. To this end, we need to construct a polytope that is feasible iff a menu is pessimistically valid for each $i\in N$. Consequently, for each $i\in N$, we introduce variables $x_o^i$, $w_o^i$ and $y_j^i$ where $o\in A_i$ and $j\in N^i$. As before, we have that $N^i=\{j\in N\setminus\{i\}|\exists o\in A_i\cap A_j\}$. We also now have $q_o$ as a variable for every $o\in O$. Further, we simulate the amount of food within the set from option $o$ (for agent $i$) $x_o^iq_os_o$ by the variable $w_o^i$. We shall introduce constraints that ensure $w_o^i=x_o^iq_os_o$.

As before, we set the range of the values of $x^i_o$ and $y^j_o$ to be only $\{0,1\}$, as they are selection variables. In order to choose the range of $q_o$ and subsequently $w_o$ we make the following observation.

\begin{observation}\label{obs:capped-amt}
    Given a \Instance{} instance $I$, under a minimum sized optimistically or pessimistically valid menu $M$, we have that $q_o(M)\leq \lceil\frac{n}{s_o}\rceil$. As a result, under $M$, the amount of food from any fixed option $o\in O$ is $q_o(M)s_o<n+s_o$.
\end{observation}

Consequently, it is sufficient to set the range of values for  $q_o$ as $\{0,1,\cdots, \lceil \frac{n}{s_o}\rceil \}$ and for $w_o$ to be $\{0,1,\cdots,n+s_o\}$. Both of these values are  determined by by the instance. Unfortunately, we cannot now simply use the previous ILP as is. Firstly, the previous ILP is focused on one agent. To this end, we now replace variables of type $x_o$ and $y_j$ by $x_o^i$ and $y_j^i$ for each $i\in N$. Further, observe that  the term $x_oq_os_o$ in the previous ILP is no longer linear now as $q_o$ is a variable. We need a way to set $w_o$ to $q_os_o$ if $x_o=1$ and set $w_o=0$ otherwise. To this end,  we introduce three new linear constraints \eqref{eq:wo-upper1}-\eqref{eq:wo-lower}. If $x_o=1$, then  \eqref{eq:wo-upper1} and \eqref{eq:wo-lower} together ensure that $w_o=q_os_o$, and \eqref{eq:wo-upper2} puts no restriction of $w_o$ as $(n+s_o)$ is a large number. Otherwise if $x_o=1$, then  $w_o=0$ due to \eqref{eq:wo-upper2}.
\begin{align}
\text{minimize \quad}                               &\displaystyle\sum\limits_{o\in O} q_os_o\nonumber\\
\text{subject to \quad}                             \nonumber\\
       \displaystyle\sum\limits_{o\in A_i} w_o^i    &\geq 1+ \displaystyle\sum\limits_{j\in N^i}y_j^i                   && i\in N\label{eq:hallinLP}\\ 
       w_o^i                                        &\leq q_os_o                                                        &&i\in N, o\in A_i\label{eq:wo-upper1}\\
       w_o^i                                        &\leq x_o^i(n+s_o)                                                  &&i\in N, o\in A_i\label{eq:wo-upper2}\\
       w_o^i                                        &\geq q_os_o -(1-x_o^i)(n+s_o)  ~~~                                 &&i\in N, o\in A_i\label{eq:wo-lower}\\
       y_j^i                                        &\geq x_o^i                                                         &&i\in N, j\in N^i, o\in A_i\cap A_j\label{eq:option-agent-feas1}\\
       y_j^i                                        &\leq\displaystyle\sum\limits_{o\in A_i\cap A_j} x_o^i              && i\in N,j\in N^i\label{eq:option-agent-feas2}\\
       x_o^i                                        &\in \{0,1\}                                                        &&i\in N,  o\in A_i\nonumber\\
       w_o^i                                        &\in \{0,1,\cdots,(n+s_o)\}                                         &&i\in N, o\in A_i\nonumber\\
       y_j^i                                        &\in \{0,1\}                                                        &&i\in N, j\in N^i\nonumber\\
       q_o                                          &\in \left\{0,1,\cdots,\big\lceil \frac{n}{s_o}\big\rceil\right\}   &&o\in O\nonumber
\end{align}

We defer a discussion of the correctness and the size of the ILP to \cref{app:complexity}. The correctness is shown in two parts: we first show that $w_o^i$ is always equal $x^i_oq_os_o$ for every agent $i\in N$ and every option $o\in A_i$. We then show that ILP above is always feasible for every menu planning instance.

\subsection{Parameterized Algorithm for Minimum Sized Optimistic Menus}\label{subsec:fpt}

In this section we design a FPT algorithm to find a minimum size menu that is optimistically valid parameterized by the number of agents.  Recall from \cref{obs:capped-amt} that under any valid menu for any option $q_o(M)\leq n$. Thus, we need to check at most $n^m$ valid menus to find a minimum size one. Consequently, there is a trivial $O(n^m)$ time algorithm. We present an algorithm that significantly improves this running time by reducing our problem to the problem of finding maximum weight partition. Towards this first we define maximum weighted partition. A tuple of subsets $(S_1, S_2, \dots S_k)$ of a set $N$ is a partition of the set $N$ if
$\cup_{i \in [k]} S_i = N$ and $S_i \cap S_j =\emptyset$ for each pair $i,j \in [k]$, and $i \neq j$.
The \textsc{Maximum Weighted Partition} problem \cite{bjorklund2009set} is defined as follows:

\begin{center}
\begin{tabularx}{1.0\columnwidth}{l}
\toprule
\multicolumn{1}{c}{\textbf{\textsc{Maximum Weighted Partition}}} \\
\midrule
\parbox[t]{0.95\columnwidth}{\textbf{Given:}  A universe $\mathcal{N}$, a set family $\mathcal{F}$ over $\mathcal{N}$, a non-negative integer $k\in \mathbb{Z}_+$, and functions $f_1, f_2, \dots f_k$, and a target value $B\in \mathbb{R}$ }\\
[1em]

\parbox[t]{0.95\columnwidth}{\textbf{Question:} Does there exist a partition $(S_1, S_2, \dots S_k)$ of $\mathcal{N}$ in to $k$ parts such that each $S_i\in \mathcal{F}$ and the sum $\sum_{i \in [k]} f_i(S_i)\geq B$?}\\
\bottomrule
\end{tabularx}
\end{center}

\begin{algorithm}[t]
\caption{Parameterized Algorithm via Maximum-Weight Partition}\label{alg:fpt}
\KwIn{Menu Planning instance $ \langle N, O, (A_i)_{i\in N}, (s_o)_{o\in O} \rangle$}
\KwOut{An optimistically valid menu $M$}
 Let $size = mn$\;
\For{$k \leftarrow 1$ \KwTo $m$}{
    \ForEach{$o_j \in O$}{
        Let $N_j \leftarrow \{ i \in N : i \text{ accepts } o_j \}$; \hfill \tcp{Define feasible types of agents}
    }
    Define $\mathcal{F} \leftarrow \{\, N' \subseteq N_j : N' \neq \emptyset,\ j \in O \,\}$\;
    \ForEach{$N' \in \mathcal{F}$}{
        Let $O' \leftarrow \{\, o_j \in O : N' \subseteq N_j \,\}$ and
        choose
       $ o' \in \arg\min_{j \in O'} \{|N'| \bmod s_j\}$\;
        Set $f(N') \leftarrow -(|N'| \bmod s_{o'})$;\hfill \tcp{Define weight for each type}
    }

    \tcp{Solve maximum weight partition}
    Compute a partition $\mathcal{X}=(X_1,\dots,X_k)$ of $N$ with $X_i \in \mathcal{F}$ maximizing
    $\sum_{i=1}^{k} f(X_i)$\;

    if $size > $ the weight of $\mathcal{X}$, then $size=$ weight of $\mathcal{X}$\;
}

Let $(X_1,\dots,X_{k^*})$ be a maximum weight partition with weight $size$\;

\ForEach{$z \in [k^*]$}{
    Let $O_z \leftarrow \{\, o_j \in O : X_z \subseteq N_j \,\}$ and
    choose
    $o_z \in \arg\min_{j \in O_z} \{|X_z| \bmod s_j\}$\;
    Set $q_{o_z}(M) \leftarrow \left\lceil \frac{|X_z|}{s_{o_z}} \right\rceil$; \hfill \tcp{Define the menu}
}

\ForEach{$o \in O \setminus \{o_z : z \in [k^*]\}$}{
    Set $q_o(M) \leftarrow 0$\;
}

\Return $M$\;
\end{algorithm}


\noindent
\textbf{Algorithm Overview.} We use the variable $k$ to denote how many distinct options are in a menu. For each integer $k\in [m]$, we reduce the given instance of \Instance{} to an instance of \textsc{Maximum Weight Partition}. We use the set of agents as the set of elements. We generate the family $\mathcal{F}$ to be the non-empty subsets of agents that all find an option to be acceptable. For each such set $N' \in \mathcal{F}$, we define the weights to be $f(N')=-(|N'| \mod s_o')$ where $o'$ is an option that minimizes this value among $\cap_{i\in N'}A_i$. We find a maximum weight partition of the agents. For each part (a subset of agents) of it, we order enough of an option that is acceptable to this subset of agents and that has minimum waste. Finally, we output the minimum size  menu over all $k \in [m]$. 

\begin{restatable}{proposition}{optfpt}\label{th:fpt}
 Given a \Instance{} instance $I$, let $s_{max}$ denote the maximum serving size of any option, i.e., $s_{max} =\max_{o \in O} s_o$ and $m=|O|$.
    \cref{alg:fpt} computes a minimum sized optimistically valid menu in $O^*(2^nm^3s_{max})$ time.  
\end{restatable}

\section{Polynomial-time Algorithms for Structured Instances}\label{sec:algorithms}

We now complement our complexity results by showing that whenever the conditions for them do not hold, we can, in fact, find minimum sized valid menus in polynomial time. 
Our results throughout this section can be combined to give the following theorems.


\begin{theorem}[Classifying Tractable Instances for Optimistic Consumption]
Finding a minimum sized optimistically valid menu is 
\begin{enumerate}
    \item[a.] NP-complete even for instances with $s_o=3$ for each $o\in O$ and
    \item[b.] in P for instances with $s_o\leq 2$ for each $o\in O$ or those with laminar acceptability relations.
\end{enumerate}
\end{theorem}

\begin{theorem}[Classifying Tractable Instances for Pessimistic Consumption]
Finding a minimum sized pessimistically valid menu is 
\begin{enumerate}
    \item[i.] NP-complete even for instances with $s_o=1$ for each $o\in O$ and
    \item[ii.] in P for instances with laminar acceptability relations
\end{enumerate}
\end{theorem}

\begin{remark}[Tightness of Intractability]
    Both our hardness reductions \cref{thm:opt-menu-npc} and \cref{thm:pess-menu-npc} create instances where for any two agents $i,j\in N$, all three of the following sets may be non-empty: $A_i\setminus A_j$, $A_j\setminus A_i$ and $A_i\cap A_j$. When at most two are non-empty, we show that the acceptability sets must be {\em laminar} (formally defined in \cref{subsec:laminar}). Further, the reduction in the proof of \cref{thm:opt-menu-npc} for optimistic consumption, constructs an instance with $s_o= 3$ for each option $o\in O$.  In contrast, when $s_o\leq 2$ a minimum sized optimistically valid menu can be found in polytime via \cref{alg:2sizedOptMenu}. Consequently, our computational results are tight in this respect. 
\end{remark}

\subsection{Limited Serving Sizes}\label{subsec:smallplates}
We now present our algorithm for $s_o\leq 2$. Note that this is a fairly realistic assumption as many restaurants only offer options for a single person or to be split between two. The algorithm  constructs a graph with weights on edges to decide how to choose which agents should share an option with $s_o=2$ and which agents require an individual option.

\begin{algorithm}[h!]
  \KwIn{ Menu Planning instance $ \langle N, O, (A_i)_{i\in N}, (s_o)_{o\in O} \rangle$ where $s_o\in \{1,2\}$ for all $o\in O$}
  \KwOut{An optimistically valid menu $M$}
  Initialize $M \gets \emptyset$\;
  Create graph $G=(V,E)$ where for each $i\in N$, add vertex $x_i$\\
  For each $o\in O$ and $i\in N$ s.t. $s_o=1$ and $o\in A_i$, add node $y_o^i$\\
  For each $i,j$ s.t. there exists $o\in A_i\cap A_j$ with $s_o=2$, add edge $(x_i,x_j)$ with weight $w(x_i,x_j)=2$\\
  For each $i$ and $o\in A_i$ with $s_o=1$, add edge $(x_i,y_o^i)$ with weight $w(x_i,y_o^i)=1$\;
  Let $\mu$ be a max weight matching on $G$ under $w$\;
  \For{each $(x_i,x_j)\in \mu$}{
        Choose $o\in A_i\cap A_j$ such that $s_o=2$\;
        $M\gets M\cup \{o\}$\;
  }
  \For{each $(x_i,y_o^i)\in \mu$}{
        $M\gets M\cup \{o\}$\;
  }
  \For{each unmatched $x_i$ under $\mu$}{
        Choose $o\in A_i$\;
        $M\gets M\cup \{o\}$\;
  }
  Return $M$\;
    
   \caption{Min Sized Optimistic Menus when $s_o\leq 2$} \label{alg:2sizedOptMenu}
\end{algorithm}

\begin{restatable}{proposition}{smallplates}\label{thm:small-plates}
        Given a \Instance{} $I=\langle N,O, (A_i)_{i\in N},(s_o)_{o\in O},\rangle$ where for each $o\in O$, $s(o) \in \{1,2\} $,  \Cref{alg:2sizedOptMenu}  finds a minimum sized optimistically valid menu in polynomial time. 
\end{restatable}

\begin{proof}[Proof Sketch]
     Recall that under optimistic consumption, it is enough to find a menu such that there is at least one way of partitioning the food order among the agents s.t. every agent finds something acceptable to eat. If each option has $s_o = 1$, a minimum sized optimistically valid menu can be constructed by adding any one acceptable option for each agent $i\in N$. However, when some options have a serving size of $2$ (or more), this simplistic approach may lead menus of much larger size than the minimum, as two people may be able to share the same option of larger size. To this end, we use a matching based approach in \cref{alg:2sizedOptMenu} to find minimum sized valid menus.
\end{proof}

\subsection{Structured Acceptability Relations}\label{subsec:laminar}
In \cref{subsec:smallplates}, we show intractability of computing min sized menus under both optimistic and pessimistic consumption for instances where given there exist  pairs of agents $i, j\in N$, s.t. all three of $A_i\setminus A_j$, $A_j\setminus A_i$ and $A_i\cap A_j$ are non-empty. We show that if at most two of these sets are non-empty, we can in fact find min sized valid menu under either consumption model in polynomial time. We now discuss three different structural assumptions which emerge when we only allow at most two of the three aforementioned sets to be non-empty for any pair of agents. 

 \textbf{Laminar Acceptability Relations.} This is the most general structured setting that allows for tractability. 
Here, we assume that for any $i,j\in N$ where $|A_i|\leq |A_j|$, it must be that $A_i \cap A_j$ is either $A_i$ or $\emptyset$. Further, we assume that there exists  some $i^*\in N$ who finds all options acceptable, that is, $A_{i^*}=O$. 
Without this assumption, we can break the given instance into multiple independent instances. We can choose each maximal acceptability set $A^v$, where for each $i\in N$ either $A_i\subseteq A^v$ or $A_i\cap A^v=\emptyset$, and all its subsets separately. As a result, we assume without loss of generality that there exists an agent who finds all options acceptable.

\begin{restatable}{lemma}{laminarstructure}\label{lem:laminar-structure}
    Given an \Instance{} instance $I$,  for every $i,j\in N$ at most two of $A_i\cap A_j$, $A_i\setminus A_j$ and $A_j\setminus A_i$ are non-empty if and only if $I$  has laminar acceptability relations.
\end{restatable}

With the assumption that there is a unique maximal set, laminar sets can be represented as a forest. Given a collection of distinct acceptability sets $A^1,\cdots, A^k$ that are laminar, there exists a tree with a vertex for each distinct set. We shall assume that this tree is rooted at the vertex corresponding to the acceptability set $A^v=O$.  We call this the laminar containment tree where $A^v$ is an ancestor of $A^{v'}$ if and only if $A^{v'}\subsetneq A^v$. 
In order to find min sized menus, we use an algorithm identical acceptability relations as a subroutine. 

\textbf{Identical Acceptability Relations.} The simplest way of having a structure is to have all agents have the same acceptability relations. That is, for all $i,j\in N$, we have that $A_i=A_j$. Consequently, we denote an instance of Identical \Instance{} by $\langle N, A, (s_o)_{o\in A}\rangle$. Observe that we do not impose any restriction on the serving sizes of options. Despite this, we find positive results for finding both min sizes optimistically and pessimistically valid menus under identical acceptability relations. Identical acceptability relations capture settings where food needs to be ordered for a homogeneous group, such as at a religious gathering or a child's birthday party. We can generalize this setting further, while still retaining some structure.


\textbf{Chained Acceptability Relations.}  Here, we assume that the acceptability relations of the agents form a sequence of subsets. That is, for any $i,j\in N$ where $|A_i|\leq |A_j|$, it must be that $A_i\subseteq A_j$. Chained acceptability relations capture dietary restrictions based on specific types of ingredients such as vegan, vegetarian, pescetarian, no red meat and unrestricted. 
Note that the laminar setting is a generalization of the chained setting.

We provide  \cref{alg:IdentOptMenu}, in \cref{app:laminar}, to find a minimum sized valid menu under identical acceptability relations. Note that under identical acceptability relations every optimistically valid menu is pessimistically valid. We use this as a subroutine to find minimum sized valid menus. Our main results in the section are as follows. It is achieved via dynamic programming based solutions to finding minimum sized menus. As is typical with dynamic programming, the algorithms return minimum possible sizes of valid menus, but the actual menus can be found analogously. 

\begin{restatable}{proposition}{laminaropt}\label{thm:laminar-opt}
   Given a \Instance{} instance $I$ with laminar acceptability relations, there exists a polynomial-time algorithm that finds the minimum size of a optimistically valid menu $M$.
\end{restatable}

\begin{proof}[Proof Sketch]
    Laminar sets can be represented by a containment tree where if $A^i\subsetneq A^j$, then $A^j$ must be an ancestor of $A^i$ in the containment tree. Let $A^1,\cdots, A^k$ be the distinct acceptability sets labeled in a pre-order traversal of the containment tree. That is, if $A_i\subseteq A_j$ then $i\geq j$. The agents with acceptability set $A^i$ can be fed from a combination of items across $A^i$ or a combination of one child alone. Consequently, to find min sized menus, we need to determine the best combination of agents being fed from each acceptability set. To this end, for each internal node $A^i$, we need to decide how many of the corresponding agents eat from $A^i$ and how many eat from each child of $A^i$. 

    To avoid going over all possible combinations here, we use the tree structure to build these combinations sequentially. We give a dynamic program in \cref{alg:LaminarOptMenu} that uses the pre-order traversal and is filled starting from leaves to root. Each node $v$ other than the root inherits a ``debt" from either its previous sibling (according to the label) or its parent (if no sibling comes earlier in the labeling), but not both. Now $v$ needs to decide how much of this debt it will keep for the subtree rooted at $v$, and how much to pass of onto the next sibling, if any. This is done by going over all the values of the debt passed on and selecting the minimum. For one specific value,  the debt retained plus $n_v$ is the number of agents that need to be fed from the subtree rooted at $v$. Of this number, $v$ now decides how many to feed from the set $A^v$ and passes off the remaining to the first child (according to the labelling), if any. This is again done by going over all the values of the debt passed on and selecting the minimum.
\end{proof}

We now give an overview of our algorithm to find minimum sized pessimistically valid menus for the laminar setting. 

\begin{restatable}{proposition}{laminarpess}\label{thm:laminar-pess}
   Given a \Instance{} instance $I$ with laminar acceptability relations, there exists a polynomial-time algorithm that finds the minimum size of a pessimistically valid menu $M$.
\end{restatable}


\begin{proof}[Proof Sketch]
    We refine the characterization of pessimistically valid menus for laminar acceptability relations in \cref{lem:LaminarPess}. Here, we show that a menu is pessimistically valid, if and only if the amount of food from each leaf of the laminar containment tree is enough to feed all agents whose acceptability set contains this leaf vertex. As a result, it is sufficient to only feed agents from these {\em minimal acceptability sets}. 
    Based on this, we present \cref{alg:LaminarPessMenu} (in \cref{app:laminar}) to compute a minimum sized pessimistically valid menu. The algorithm proceeds by considering the laminar containment tree and the $t$ leaves $A^{v_1},\cdots, A^{v_t}$. For each leaf $v_j$ it adds the minimum size pessimistic menu for feeding $n_{v_j} + \sum_{v:A^{v_j}\subsetneq A^{v}} n_{v}$ from $A^{v_j}$ using the algorithm for identical acceptability relations as a subroutine. The combined sizes for each leaf is returned.
\end{proof}

\begin{example}[Min sized valid menus for Laminar Acceptability Relations]\label{ex:laminar}
    Consider the following instance where  $n=12$ and $m=4$. Let $A_1=\{o_1\}$, $A_2=A_3=A_4=\{o_2\}$, $A_5=\{o_1,o_2\}$, $A_6=\{o_3\}$ and $A_7=\cdots=A_{12}=\{o_1,o_2,o_3,o_4\}$. Further, set serving sizes as follows: $s_{o_1}=2$, $s_{o_2}=s_{o_3}=3$ and $s_{o_4}=1$. Observe that any menu that does not contain $\{o_1,o_2,o_3\}$ will not be valid. Here the minimum sized optimistically valid menus will have size $n=12$. Some examples are:  $M_1=\{o_1,o_2,o_3,o_4,o_4,o_4,o_4\}$, $M_2=\{o_1,o_1,o_1,o_2,o_3\}$ and $M_3=\{o_1,o_2,o_2,o_3,o_4\}$.
    Further, menu $M_4=\{o_1,o_2,o_3,o_3,o_3\}$ will also be optimistically valid but has size $14$. 

    Recall that we need enough food from each minimal acceptability set to be pessimistically valid. Observe that exactly $8$ agents find $o_1$ acceptable, $10$ find $o_2$ acceptable and $7$ find $o_3$ acceptable. Consequently, the unique pessimistically valid menu is $M_5=\{o_1,o_1,o_1,o_1,o_2,o_2,o_2,o_2,o_3,o_3,o_3\}$ which has size $29$.
\end{example}

\begin{remark}
    Given a graph $G=(V,E)$ where the neighborhoods of vertices are laminar, can find a minimum sized maximal matching in polynomial time.     
\end{remark}



\section{Waste of  Pessimism}\label{sec:wop}

We now try to measure the difference in the menu sizes needed if pessimistic consumption is assumed but optimistic consumption is practiced. Informally, we are trying to measure  how much extra food needs to be ordered to have a pessimistically valid menu vs an optimistically valid menu. We call this the {\em waste of pessimism} (WOP). Recall that $\Lambda(I)$ and $\Gamma(I)$ denote the spaces of optimistically and pessimistically valid menus, respectively.

\begin{definition}[Waste of Pessimism]
    Given a \Instance{} instance $I=\langle N, O, (A_i)_{i\in N},(s_o)_{o\in O}\rangle$, the waste of pessimism is the ratio of the minimum sized pessimistically valid menu under $I$ and that of the minimum sized optimistic menu
    \[WoP(I)=\frac{\min_{M \in \Gamma(I)} \sum_o q_o(M)s_o }{\min_{M\in \Lambda(I)} \sum_o q_o(M)s_o }.\]
\end{definition}

\begin{example}[Waste of Pessimism] We now compute the waste of pessimism of instances discussed in previous examples. 

\begin{enumerate}
    \item Consider \cref{example:basic} with $4$ agents $N=\{1,2,3,4\}$ and two options $O=\{o_1,o_2\}$ where $s_{o_1}=3$ and $s_{o_2}=1$, and   $A_1=A_2=A_3=\{o_1\}$ and $A_4=\{o_1,o_2\}$. Here, $\{o_1,o_2\}$ is an minimum sized optimistically valid menu. On the other hand $\{o_1,o_1\}$ is the min-sized pessimistically valid. The waste of pessimism for this instance is $6/4=3/2$.
    \item Consider \cref{ex:laminar}. Recall that the minimum sized optimistically valid menu has size $n=12$ and the minimum sized pessimistically valid menu has size $29$. Here, the waste of pessimism is $29/12$.
\end{enumerate}

    \end{example}

Observe that the instance described in \cref{example:basic} is one with chained acceptability relations, whereas the one in \cref{ex:laminar} is one with laminar acceptability relations.
We now prove tight bounds on the Waste of Pessimism for structured acceptability relations. Our bounds show that this stark contrast in waste of pessimism between laminar and chained scales  as $n$ increases.

\begin{theorem}\label{thm:wop-laminar}
    Given a \Instance{} instance $I$, we have that  
    \begin{enumerate}
        \item if $I$ has laminar acceptability relations, $WoP(I)\leq \frac{n+4}{4}$,
        \item if $I$ has chained acceptability relations, $WoP(I)\leq 2$, and
        \item if $I$ has identical acceptability relations, $WoP(I)=1$.
    \end{enumerate}
\end{theorem}

\begin{proof}[Proof Overview]
    Fix a \Instance{} instance $I=\langle N,O,(s_o)_{o\in O},(A_i)_{i\in N}\rangle$ with laminar acceptability relations. We first consider the laminar case. Given instance $I$, we shall create an alternate instance $I'$, also with laminar acceptability relations with the same number of agents  and minimal acceptability sets, but which will always satisfy $WoP(I)\leq WoP(I')$ and then reason about its maximum possible waste of pessimism value. We will prove this bound as follows:

    \begin{enumerate}
        \item[i.]   Create an alternate instance $I'$ with laminar acceptability relations, with the same number of agents and the same number of minimal acceptability sets. 
        \item[ii.]  Prove that $WoP(I)\leq WoP(I')$.
        \item[iii.] Find a worst case choice of serving sizes and number of minimal acceptability sets for WoP
        \item[iv.]  Demonstrate an instance where this bound is tight. 
    \end{enumerate}

    Essentially, the alternate instance created is such that its corresponding laminar containment tree is a star graph that has the same number of leaves as that of $I$. Under $I'$, the root and the leaves of the two tree correspond to the same sets as those of the root and leaves under $I$. Further, under $I'$, exactly one agent finds has an acceptability set corresponding to each leaf. All other agents have acceptability set corresponding to the root. From \cref{lem:LaminarPess}, the amount of food needed from each leaf under $I'$ can only increase now. We use this to show that $WoP(I)\leq WoP(I')$. As $I'$  significantly  structured, we can argue about its worst case bounds. We show that the worst case $WoP$ happens when the number of leaves is $n/2$ and the minimum serving size of any option in a minimal acceptability set has size at least two.
\end{proof}    

\section{Conclusions and Discussion}

 We have introduce the problem of menu selection and study its computational complexity. We characterize the space of instances that allow for polynomial-time algorithms for computing minimum sized optimistically valid menus. For the pessimistic case, we show that for some realistic settings, minimum sized pessimistically valid menus can in fact be computed efficiently.  We introduce a price of anarchy analog for this setting, \emph{the waste of pessimism}. We find tight upper bounds on the price of anarchy for the case of laminar, chained and identical acceptability relations. 

 Our paper opens many interesting avenues for future work. Firstly, there is no known tight bound on waste of pessimism outside of laminar instances. We conjecture that the upper bound for laminar should in fact extend to {\em all} instances. Additionally, there are many other directions that can be pursued within this model. One important direction is to develop approximation algorithms which find valid menus of size at most $\alpha\geq 1$ times the minimum sized valid menu. Similarly, it would be interesting to explore how to build ``minimal'' valid menus  and understand what guarantee they achieve on the minimum sized valid menu. 

Our paper offers a new perspective on the literature on online matching problems (e.g., \cite{KVV90,HTW24,EIV23a}). In online bipartite matching, one side of the graph is known in advance, while vertices on the other side arrive sequentially with their incident edges, and the algorithm must irrevocably decide whether to match each arrival to an available neighbor. The objective is to maximize the size or weight of the matching, a model that captures many real-time allocation settings and has been widely studied under adversarial, random-order, and stochastic arrivals. Our problem can be seen as a preceding decision problem: selecting a subset of the known side under capacity or cost constraints. While we adopt optimistic and pessimistic perspectives, it would be interesting to incorporate ideas from online matching that exploit stochastic information about arrivals or consumption, e.g., kidney matching \cite{dickerson2021allocation,dickerson2012optimizing} or online recommender systems settings \cite{aird2024dynamic}.

 
 Another direction would be to introduce additional objectives or constraints such as costs over options. Alternately, there may be complementary options where a protein and/or vegetable based option, such as a burger, steak or salad, needs to be paired with a carbohydrate based option like rice or bread. 
 Another constraint may be that different agents have different hunger levels and need different amounts of food. For instance, adults need one portion but kids may only need half a portion. Our DP based algorithms extend to this case in a straightforward manner, as we simply need to adjust the target goal. 
 
 Finally, it would also be interesting to introduce agent preferences over acceptable options. This allows us to pursue objectives such as Pareto optimality, maximum welfare and even fairness in our choice of menu. Potentially, this may allow for lower pessimistic risk, if agents only eat their most preferred options from a given menu. As a result, our work opens many intriguing possibilities of further work on menu selection with applications to other collective choice settings \cite{aziz2020participatory,aleksandrov2020online}.

\newpage

\bibliographystyle{ACM-Reference-Format}
\bibliography{anarchy}

\appendix


\section{Related Work}\label{app:relwork}

The problem of minimizing food waste and its various causes have been studied from a variety of systemic angles \cite{pandey2021food,moraes2021systematic,grandhi2016waste,panda2019minimizing,matharu2022efforts}. Poor planning and overpurchasing has been consistently cited as a cause of food waste. In our paper, we take a complementary angle to the previous work, where we look at an algorithmic approach to choosing the most efficient menu to feed all agents. We introduce a novel framework, and to the best of our knowledge, no other work has studied an algorithmic model where every agent must find something to eat from a collective, waste minimizing menu.

Our model has connections with two problems in computational social choice: multi-winnner voting (see e.g.,~\citet{lackner2023multi}) and fair division (see \citet{amanatidis2023fair} for a survey). 
Similar to multi-winner voting, we need to select one menu from a given set of options which will then be shared among the agents. Unlike voting problems, we assume an upper bound  (via serving sizes) on how much of an option can serve an agent. Once a menu has been selected, one can consider ways to allocate the food among the agents. Unlike standard allocation settings, we can choose the options that will then be implicitly shared across agents.

Recently, \citet{fish2025stable} studied a related problem within a public good setting. 
They look for a stable selection of public goods such that each item is used by at least $s$ agents and the menu is not blocked by more than $t$ agents. In contrast to our model, as \citet{fish2025stable} consider public goods, where each option effectively has unlimited serving sizes as opposed to limited serving sizes in our model. Further, they do not require all agents to have at least one acceptable option in the menu, and some agents may go hungry. These assumptions are common when studying public goods \cite{fain2018fair,banerjee2023proportionally,kroer2025computing,garg2021fair}. However, they need not hold in many day-to-day settings of ordering food. 

Our problem, under the pessimistic consumption model, can be viewed as a natural decision problem that \textit{precedes } online matching problems (see, e.g., \cite{KVV90,HTW24,EIV23a}). Online matching has been well studied since the work of \citet{KVV90}.  
We discuss this connection in more detail in the Conclusions. Bipartite one-one matchings are also equivalent to house allocation with approval preferences.  This perspective also connects our work to control problems. 


\paragraph{Relations To Control and Optimization} 
Our work also extends existing work on control problems within social choice (see \citet{CKN+2025control} for a survey). Within social choice, control problems typically explore computational complexity of adding or removing agents/items to ensure the existence of solutions that satisfy certain properties \citep{ASW16a,GNS+24capacity,BBH+21bribery,CC23optimal}. Our work differs from existing research in this space on two fronts. Firstly, the work in this space has been agnostic to the choice of mechanism or algorithm. In contrast, we introduce the question of controlling with a specific algorithm in mind. Secondly, prior work on control and social choice has not considered such ``house allocation'' type problems where agents have to each be matched to one item, but rather either looked at fair division \citep{ASW16a,BBD+2025resolve} or two-sided matchings \citep{BBH+21bribery,BCL+22capacity,BCL+23capacity,CC23optimal,GNS+24capacity}. The majority of the work in this space is also an \textit{optimization problem}: find the minimum number of items/agents to add or remove so desirable outcomes can exist. 

Covering problems are another family of optimization problems that relates to our problem under optimistic consumption. Problems like set cover, vertex cover, domination number are well known to be NP-hard \cite{garey1979computers}. Our model is a generalization of these problems. For example, when the serving size of each option is equal to the number of people who find it acceptable, then a minimum size optimistically valid menu is the same as a minimum size set cover. In addition to covering every agent by an option, we also restrict the ability of an option to simultaneously cover multiple agents based on its serving size. Consequently, multiple copies of an option may be needed to feed all the agents. 
Our goal of finding smallest menus that feed all agents can be also recast as an optimization problem where the number of agents that are fed are captured by a subadditive function. We wish to find the smallest menu for which our objective function has value equal to the number of agents. Subadditive functions and their subtypes have been widely studied in the context of welfare maximization for item allocation \cite{feige2006maximizing,dobzinski2024constant,barman2021approximating,barman2024sublinear} and combinatorial procurement auctions \cite{badanidiyuru2012optimization,badanidiyuru2012sketching,chekuri2011submodular}. Typically, there is no minimum size requirement on the set selected. Rather, in work on item allocation, there may be an upper bound on the  cardinality of the set allocated to an agent \cite{badanidiyuru2012optimization,biswas2018fair}. Here, the optimization objective is the value of the function, not the sets. In contrast, we have a given value of our function and we wish to pick the smallest set to achieve this value.

The term {\em menu selection} has been used in the past to refer to other research problems, especially in human-computer interaction (HCI). In the field of HCI, 
menu selection is concerned with what options to show to a user and how should they be presented in a drop-down menu or list \cite{karat1986comparison,shneiderman1986designing,arend1987evidence,wilson2010pressure}. Outside of computer science, food researchers have also heavily studied the effect of menu selection on customer experience and restaurant sales \cite{jung2014study,kang2011study,peters2020factors,ozdemir2015menu}. Here, menu selection refers to not only the dishes on offer but the fonts, designs, options, and other interface details.

\section{Omitted Material from \cref{sec:characterizations}}\label{app:charac}

\subsection*{Optimistic Consumption}

\optcharac*
\begin{proof}
    Fix an instance $I$ and a menu $M$. Recall that the optimistic consumption is precisely the size of the maximum cardinality matching in $G_M$. Thus the equivalence of statements $(1)$ and $(2)$ is straightforward. 
    We now show the equivalence of these statements with there existing a consumption sequence $\omega$. 
    
    Let $M$ is optimistically valid. Choose $x$ to be a maximum cardinality matching in the consumption graph $G_M$. Consequently, we have that for every agent $i\in N$, there exists $o\in A_i$ s.t. $x_{i,o}=1$.  For each $k\in [n]$, let $\pi$ be an arbitrary ordering over $N$ and set $\omega(k)=o$ s.t. $x_{\pi(k),o}=1$.  As $x$ is a matching, we have that for each option $o\in O$, $\sum_{i\in N}x_{i,o}\leq q_os_o$.  Thus, for every agent $k\in [n]$, it must hold that the remaining menu $M^k$, after previous agents have picked, must satisfy $M^k|_{A_{\pi(k)}}$. Specifically, it must hold that $o\in M^k|_{A_{\pi(k)}}$ where $o$ is s.t. $x_{\pi(k),o}=1$. Thus, $\omega$ is indeed a consumption sequence of $M$.

    For the converse, let there exist a consumption sequence $\omega$ s.t. for each $k\in [n]$, $\omega\neq \emptyset$. Let $\pi$ be a corresponding ordering over agents. Consider the matching $x$ where for each $k\in [n]$, $x_{\pi(k),\omega(k)}=1$ and $x_{\pi(k),o}=0$ for all $o\neq \omega(k)$. As $\omega$ is a consumption sequence $\omega(k)\in A_{\pi(k)}$ and $q_{\omega (k)}(M)>0$. Further, the number of agents matched to option $o$:  $\sum_{i\in n}x_{i,o}=|\{k\in [n]|\omega(k)=o\}|\leq q_os_o$. Therefore, $x$ is indeed a matching of the consumption graph $G_M$ and has size $n$. Therefore, $M$ must be optimistically valid. 
\end{proof}

\subsection*{Pessimistic Consumption}


\pesscharac*
\begin{proof}
    We shall first prove equivalence of pessimistic validity to the condition on pessimistic risk for each agent. 

    \textbf{Equivalence of statements $(1)$ and $(2)$.}  
    Assume $M$ is pessimistically valid. Thus every maximal matching in the consumption graph $G_M$ must have size $n$. Fix agent $i\in N$. Consider a maximal matching where as many agents eat from $A_i$ as possible. In particular, we consider  a maximum cardinality matching from $N\setminus \{i\}$ to the multiset $M^i$ where $q_o(M^i)=q_o(M)$ for all $o\in A_i$ and $q_o(M^i)=0$ for all $o\notin A_i$. Note that the size of this matching is $\rho_i(M)$. As $M$ is pessimistically valid, this matching must not be maximal. In particular, it must hold that $\rho_i(M)<\sum_{o\in A_i} q_o(M)s_o$. Consequently, we have that for all agents $i\in N$, $\sum_{o\in A_i} q_o(M)s_o\geq 1+ \rho_i(M)$.
    
    Now it remains to prove the converse. Assume that $M$ is s.t. for all agents $i\in N$, it  holds that $\rho_i(M)<\sum_{o\in A_i} q_o(M)s_o$. Consequently, we have that for all agents $i\in N$, Fix agent $i\in N$ and an arbitrary maximal matching $x$ in consumption graph $G_M$. It suffices to prove that for $i$ and $x$ we have that $\sum_{o\in A_i} x_{i,o}=1$. 
    
    By assumption, we have that $\sum_{o\in A_i} q_o(M)s_o\geq 1$. Consequently, for at least one option in $o\in A_i$, it must be that $q_o(M)\geq 1$. By definition of pessimistic risk, at most $\rho_i(M)$ agents in $N\setminus \{i\}$ can be matched to an option in $A_i$. As we have that $\sum_{o\in A_i}q_o(M)s_o\geq 1+\rho_i(M)$, any maximal matching should not leave agent $i$ unmatched. Thus, as $x$ is a maximal matching, it must hold that $\sum_{o\in A_i}x_{i,o}=1$.
    Hence, $M$ is pessimistically valid. 

    \textbf{Equivalence of statement $(3)$ to pessimistic validity.} Let $M$ be pessimistically valid. Consider a consumption sequence $\omega$ of $M$ with a corresponding ordering $\pi$. Recall that for any $k\in [n]$, each preceding agent, under $\pi$, can consume at most one unit of an option acceptable to agent $i$.  Consequently, 
    we have that the difference between the number of units of acceptable options for agent $\pi(k)$ under $M$ and that under $M^k$ is $size(M|_{A_{\pi(k)}})-size(M^k|_{A_{\pi(k)}})\leq \rho_i$. As $M$ is pessimistically valid, we have that $size(M|A_i)=\sum_{o\in A_i} q_o(M)s_o\geq 1+ \rho_i(M)$. Consequently, $M^k|_{A_{\pi(k)}}$ must be non-empty. As a result, $\omega(k)\neq \emptyset$ for each $k\in [n]$.

    Finally, for the converse, let $M$ be s.t. for every consumption sequence $\omega$ of $M$, we have that for each $k\in [n]$, $\omega (k)\neq \emptyset$. Fix agent $i\in N$. Let $x$ be maximum cardinality matching between $N\setminus \{i\}$ and $M|_{A_i}$ . Consider an arbitrary ordering $\pi$ under which $\pi(n)=i$. That is, $i$ is the last agent to pick from $M$. Let $N^i=\{j\in N\setminus \{i\}|\sum_{o\in A_j} x_{j,o}=1\}$.  Consider a choice of $\omega$ s.t. for $\pi(k)\in N^i$  and the option $o$ for which $x_{\pi(k),o}=1$, we set $\omega(k)=o$. For all $k\in [n]$ s.t. $k\notin N^i$, choose $\omega (k)$ arbitrarily from $M^k$, whenever $M^k$ is non-empty.     
    It is straightforward to see that $\omega$ is indeed a consumption sequence for $M$. Consequently, it must hold that $\omega(n)\neq \emptyset$, by assumption. As a result, we must have that $M^n|_{A_i}$ must be non-empty. This gives us that the size of $x$ if $\rho_i=\sum_{j\neq i}\sum_{o\in A_j}x_{j,o}<\sum_{o\in A_i} q_o(M)s_o$. 
    
    Hence, we have that for each $i\in N$, $size(M|_{A_i})>\rho_i$ and thus, $M$ is pessimistically valid. 
\end{proof}


\section{Omitted Material from \cref{sec:complexity-ILPs}}\label{app:complexity}
\subsection*{Computational Complexity}
\validverify*
\begin{proof}

    Given a \Instance{} instance $I=\langle N,O, (A_i)_{i\in N},$ $ (s_o)_{o\in O},\rangle$ and menu $M$, in order to check if $M$ is optimistically valid we need to check if there is a matching $x:N\times O \rightarrow [0,1]$ where for each agent $i\in N$, $\sum_{o\in A_i}x_{i,o}=1$ and for each option $o\in O$, $\sum_{i\in N} x_{i,o}\leq q_o(M)s_o$. This can be completed in polynomial time using a bipartite maximum matching algorithm (See \citet{west2001introduction} for an overview).

    In order to compute the pessimistic consumption of $M$, we need to find a minimum sized maximal matching. The problem of finding a minimum size maximal matching is known to be NP-hard  even on 3-regular bipartite graphs (GT10 in Appendix A1.1 \cite{garey1979computers}). Consequently the problem of computing the pessimistic consumption of a given menu is NP-hard in general.

    {\em Checking validity of $M$.} To check if $M$ is optimistically valid, it is sufficient to check if its consumption graph admits an agent saturating matching. This can be done in polynomial time using a bipartite maximum matching algorithm. For pessimistic validity, we know that we cannot check the pessimistic consumption in polynomial time. Recall from \cref{lem:pess-valid} that $M$ is pessimistically valid if and only if for every agent $i\in N$, $size(M|_{A_i})\geq 1+\rho_i$. Note that, $\rho_i$ is the size of the maximum cardinality matching between $N\setminus \{i\}$ and $M|_{A_i}$. As a result, $\rho_i$ can be computed by a bipartite maximum matching algorithm as well. Hence, it can be checked in polynomial time if $M$ is pessimistically valid.  
\end{proof}


\optcomplexity*

\begin{proof}
    \noindent {\em Membership in NP.} Given instance $I$ and a size $B$, we can verify if $I$ admits a valid menu of size at most $B$, by using the menu as a certificate. It can checked in polynomial time if its size is at most $B$ and whether it is optimistically valid, by \cref{thm:verify}.\\

    \noindent {\em NP-hardness.} We give a reduction from the \textsc{3DimensionalMatching} (3DM) problem. In this problem, we are given a tripartite hypergraph $G=(X,Y,Z,E)$ where $|X|=|Y|=|Z|=k$ and for each $e\in E$ we have that $e\in X\times Y\times Z$. We need to decide if there exists a (perfect) 3DM of size $k$. This problem is shown to be NP-hard, even when each vertex is contained in at most two hyperedges \cite{garey1979computers}.  We reduce it to our setting as follows:

    Given $G=(X,Y,Z,E)$ we create a vertex agent $i_a$ for each vertex $a\in X\cup Y \cup Z$. Further, for each $e=(x,y,z)\in E$ we create an edge option $o_e$ s.t. size $s(o_e)=3$ $o_e\in A_{i_a}$ if and only if $a\in \{x,y,z\}$. Thus we have $n=3k$ agents and $m=|E|$ edges. 

    We now show that that an optimistically valid menu of size $n=3k$ exists if and only if the given 3DM instance has a 3DM of size $k$. 
    
    Suppose an optimistically valid menu of size $n=3k$ does indeed exist. Let this be $M$. Under $M$, for each $i\in N$, there exists $o\in A_i$ s.t. $q_M(o) \geq 1$. Recall that each option has serving size three. As $M$ has size $3k$, at most $k$ options must have been selected. Further, each option is contained in the acceptable set of exactly three agents, therefore, at least $k$ distinct options are needed to feed all agents. As a result, we must have that $q_M(o)\leq 1$, for all $o\in O$. Consequently, we have that the set $\mu=\{e\in E|q_M(o)=1\}$ must be a perfect matching.

    For the converse, given a perfect matching $\mu\subset E$, we consider the menu $M=\{o_e|e\in \mu\}$. As $\mu$ is a perfect matching, for each $i\in N$, there exists an option $o\in M\cap A_i$. Further, as $s(o)=3$ for each $o$, we have that one serving of $o$ serves all three agents who find it acceptable. Consequently $M$ is an optimistically valid menu.  
\end{proof}

\pesscomplexity*


\begin{proof}
   \noindent {\em Membership in NP.} Given instance $I$ and a size $B$, we can verify if $I$ admits a valid menu of size at most $B$, by using the menu as a certificate. It can checked in polynomial time if its size is at most $B$ and whether it is pessimistically valid, by \cref{thm:verify}.\\

   \noindent {\em NP-hardness.} We give a reduction from the \textsc{Exact3Cover} (X3C) problem, which is known to be NP-hard \cite{garey1979computers}. In this problem, given a universe of $3k$ elements $\mathcal{U}$ and a family of $\ell$ subsets $\mathcal{S}$, where for each $S_t\in \mathcal{S}$, we have that $|S|=3$. The objective is to determine if there exist $k$ mutually disjoint subsets which cover $\mathcal{U}$. We shall show that an X3C exists for $\langle \mathcal{U},\mathcal{S}\rangle$ if and only if a pessimistically valid menu of size $n$ exists.  Given an instance of X3C $\langle \mathcal{U},\mathcal{S}\rangle$, we create a \Instance{} instance as follows:
   \begin{itemize}
       \item For each element $e\in \mathcal{U}$, create element agent $i_e$,
       \item For each subset $S_t\in \mathcal{S}$, create three set agents $j^t_1$, $j^t_2$ and $j^t_3$ and options $o^t_1,\cdots, o^t_4$.
       \item For each $o\in O$, we set $s_o=1$
       \item For each $S_t\in \mathcal{S}$, let $S_t=\{e^t_1,e^t_2,e^t_3\}$. For each $t'\in [3]$, $A_{j^t_{t'}}=\{o^t_{t'}.o^t_4\}$. Further, add $o^t_{t'}$ to $A_{i_{e_{t'}}}$.       
   \end{itemize}

   Consequently, we create an instance with $n=3(k+\ell)$ agents and $m=4\ell$ options. For each $e\in \mathcal{U}$, $A_{i_e}$ contains exactly one option corresponding to each subset containing $e$. Note that the exact specification of which $o^t_{t'}$ is matched to which element agent is not strict. It is only needed that each element agent find a distinct option from the set acceptable. Further each $o^t_{t'}$ is acceptable to exactly two agents: one element agent and one set agent. As a result, for each element agent $i_e$ and each $o\in A_{i_e}$, there is a distinct set agent who also finds that option acceptable. Similarly, for each set agent $j^t_{t'}$ only $o^t_{t'}$ and $o^t_4$ are acceptable. Note that $j^t_{t'}$ is the only agent who finds both these options acceptable.   

   From \cref{lem:pess-valid}, we know that for a menu $M$ to be pessimistically valid we need for each $i\in N$ that $\sum_{o\in A_i} q_os_o\geq 1+\rho_i$, where $\rho_i$ is the size of the maximum cardinality matching from $N\setminus\{i\}$ to $M|_{A_i}$. For an element agent $i_e$, we have that $\rho_i=|\{o\in A_i|q_o(M)>0\}|$. Consequently, for a menu $M$ to be pessimistically valid, we need for every element agent $i_e$, it must hold that $q_o(M)\geq 2$ for some $o\in A_i$. Now consider a set agent $j^t_{t'}$. We can analogously see that if $q_{o^t_{t'}}(M)\leq 1$ AND $q_{o^t_4}(M)\leq 2$ then $M$ is not pessimistically valid for $j^t_{t'}$. Thus, for $M$ to be pessimistically valid, we need for each $t\in [\ell]$ and each $t'\in [3]$ that either $q_{o^t_{t'}}\geq 2$ or $q_{o^t_4}\geq 3$.  

   \textbf{Constructing menu from X3C.} We shall now show that an X3C exists for $\langle \mathcal{U},\mathcal{S}\rangle$ if and only if a pessimistically valid menu of size $n$ exists in the instance constructed. Let $S_{t_1},\cdots, S_{t_k}$ be an exact 3-cover. Consider the following menu $M$ where for each $t\notin \{t_1,\cdots t_k\}$, we set $q_{o^t_4}=3$ and for each $t'\in \{1,2,3\}$ we set $q_{o^t_{t'}}(M)=0$. For each $t\in \{t_1,\cdots, t_k$, we set $q_{o^t_4}=0$ and for each $t'\in \{1,2,3\}$ we set $q_{o^t_{t'}}(M)=2$. The size of $M$ is $6k+3(\ell-k)=3(\ell+k)=n$. Thus, it only remains to prove pessimistic validity.

   We first show pessimistic validity for element agents $i_e$. Observe that as $S_{t_1},\cdots,S_{t_k}$ is an X3C, for each $e\in \mathcal{U}$, there exists $t\in \{t_1,\cdots, t_k\}$ s.t. $e\in S_t$. Note that, we have $o^t_{t'}\in A_{i_e}$ for some $t'\in [3]$ and $q_{o^t_{t'}}(M)=2$. Further, as there is a unique set covering $e$, we have that $\rho_{i_e}=1$. Therefore, $M$ is pessimistically valid for each element agent $i_e$.  
   
   For set agent $j^t_{t'}$, there are two possible cases. If $t\in \{t_1,\cdots, t_k\}$, we have that $q_{o^t_{t'}}(M)=2$ and $q_{o^t_4}(M)=0$. Consequently, $\rho_{j^t_{t'}=1}$ and $M$ is pessimistically valid for $j^t_{t'}$. The remaining case is when $t\notin \{t_1,\cdots,t_k\}$. In this case, we have that $q_{o^t_4}=3$ and $q_{o^t_{t'}=0}$. Here, we have that $\rho_{j^t_{t'}=2}$ which is strictly less than $q_{o^t_4}(M)$. Thus, $M$ is pessimistically valid and has size $n$.

   \textbf{Constructing X3C from menu.} We now consider the case that there exists a pessimistically valid menu $M$ of size $n$. We now show that we can use it to construct an exact 3-cover. Recall that we need for pessimistic validity that for each set $S_t$, either $q_{o^t_4}\geq 3$ or $q_{o^t_{t'}}\geq 2$ for each $t'\in \{1,2,3\}$. We first prove that for each element agent exactly one acceptable option must be in $M$. Fix $e\in \mathcal{U}$. We know that for pessimistic validity, it must hold that for at least one $o\in A_{i_e}$, $q_{o}\geq 2$. Suppose, for contradiction, that there are two distinct options $o\in A_{i_e}$ s.t. $q_o>0$.  As the size of the menu is exactly $n$, by pigeonhole principle this leaves at least one agent for whom no acceptable option in $M$, which contradicts the fact that $M$ is pessimistically valid. Therefore, $M$ must have exactly one option s.t. $q_o(M)>0$. Specifically, for this option, $q_o(M)=2$. 
   
   For agent $i_e$ let the unique option be $o^t_{t'}$ for some $t'\in [3]$. We now show that this must mean that $q_{o^t_1}(M)=q_{o^t_2}(M)=q_{o^t_3}(M)=2$.  Suppose not, that let $q_{o^t_p}<2$ for some $p\in [3]$. Thus, for pessimistic validity, we need that $q_{o^t_4}\geq 3.$ This again means that there is more food in $M$ from options corresponding to $S_t$ than there are agents consuming it, which contradicts the pessimistic validity of $M$ which has size $n$. Therefore, we have that under $M$ for each $t\in [\ell]$ either i) $q_{o^t_4}=3$ and  $q_{o^t_1}(M)=q_{o^t_2}(M)=q_{o^t_3}(M)=0$ or ii) $q_{o^t_4}=0$ and $q_{o^t_1}(M)=q_{o^t_2}(M)=q_{o^t_3}(M)=2$. 
   
   Recall that for each element agent, there is a unique option $o$ that has $q_o(M)>0$. Consequently, for any two sets $S_t$ and $S_v$ s.t. $q_{o^t_4}(M)=q_{o^v_4}=0$, we have that $S_t\cap S_v=\emptyset.$ Therefore, it must be that the sets with $q_{o^t_4}(M)=0$ constitute an X3C of $\langle \mathcal{U},\mathcal{S}\rangle$.
   

\end{proof}

\subsection*{ILP for Minimum Sized Pessimistically Valid Menus}
We first recall the ILP for computing the minimum sized pessimistically valid menu.

\begin{align}
\text{minimize \quad}                               &\displaystyle\sum\limits_{o\in O} q_os_o\nonumber\\
\text{subject to \quad}                             \nonumber\\
       \displaystyle\sum\limits_{o\in A_i} w_o^i    &\geq 1+ \displaystyle\sum\limits_{j\in N^i}y_j^i                   && i\in N\tag{\ref{eq:hallinLP}}\\ 
       w_o^i                                        &\leq q_os_o                                                        &&i\in N, o\in A_i\tag{\ref{eq:wo-upper1}}\\
       w_o^i                                        &\leq x_o^i(n+s_o)                                                  &&i\in N, o\in A_i\tag{\ref{eq:wo-upper2}}\\
       w_o^i                                        &\geq q_os_o -(1-x_o^i)(n+s_o)  ~~~                                 &&i\in N, o\in A_i\tag{\ref{eq:wo-lower}}\\
       y_j^i                                        &\geq x_o^i                                                         &&i\in N, j\in N^i, o\in A_i\cap A_j\tag{\ref{eq:option-agent-feas1}}\\
       y_j^i                                        &\leq\displaystyle\sum\limits_{o\in A_i\cap A_j} x_o^i              && i\in N,j\in N^i\tag{\ref{eq:option-agent-feas2}}\\
       x_o^i                                        &\in \{0,1\}                                                        &&i\in N,  o\in A_i\nonumber\\
       w_o^i                                        &\in \{0,1,\cdots,(n+s_o)\}                                         &&i\in N, o\in A_i\nonumber\\
       y_j^i                                        &\in \{0,1\}                                                        &&i\in N, j\in N^i\nonumber\\
       q_o                                          &\in \left\{0,1,\cdots,\big\lceil \frac{n}{s_o}\big\rceil\right\}   &&o\in O\nonumber
\end{align}

\paragraph{Correctness of ILP} 
We need to prove that this ILP is feasible only for choices of $q_o$ for $o\in O$ that correspond to a pessimistically valid menu. To this end, we first prove that $w_o^i=x_o^iq_os_o$ under every feasible solution $z=((q_o)_{o\in O}, (x_o^i)_{i\in N,\ o\in A_i}, (w_o^i)_{i\in N, o\in A_i}, (y_j^i)_{i\in N,\ y\in N^i} )$. 
Fix agent $i\in N$ and option $o\in A_i$. 

\underline{Case 1:} $x_o^i=1$. In this case constraint \ref{eq:wo-lower} reduces to $w_o^i\geq q_os_o$. Further, as a consequence of \cref{obs:capped-amt} constraint \ref{eq:wo-upper2} becomes a trivial upper bound of constraint \ref{eq:wo-upper1} and we get that $w_o^i=q_os_o$.

\underline{Case 2:} $x_o^i=0$. In this case, constraint \ref{eq:wo-upper2} insists that $w_o^i\leq 0$. Observe that this satisfies constraint \ref{eq:wo-upper1}. Further, from \cref{obs:capped-amt} we have that $n+s_o>q_os_o$. Combining  this with $x_o^i=0$ makes constraint \ref{eq:wo-lower} trivially satisfied in this case. Thus, we must have that $w_o^i=0$ whenever $x_o^i$.

Consequently, we have that $w_o^i=x_o^iq_os_o$. As a result, our current constraints are simulating the constraints of the previous ILP for each agent $i\in N$. Thus, for a choice of $q_o$ for each $o\in O$, the ILP has a feasible solution only if the corresponding menu is pessimistically valid for every agent. 

It remains to show that the ILP is indeed always feasible for some choice of $q_o$ for an arbitrary instance. Given \Instance{} instance $I$, observe that one trivial feasible solution is when we set $q_o=\lceil \frac{n}{s_o}\rceil$. Every option has enough quantity to feed {\em all} agents. As a result, this ILP will be feasible for all instances. Hence, the ILP finds a minimum sized pessimistically valid menu.

\paragraph{Size of ILP} Observe that there are $m$ variables of type $q_o$, at most $nm$ variables of type $x^i_o$ and $w^i_o$ each and at most $n^2$ variables of type $y_j^i$. Further, there are $n$ constraints corresponding to the Hall's theorem violation constraints, at most $3mn$ constraints corresponding to $w_o^i$ simulating $x_o^iq_os_o$ and at most $n^2(m+n)$ constraints corresponding to $y_j^i$ being $1$ if and only if $x_o=1$ for some $o\in A_i\cap A_j$. Finally, there are at most $2nm+n^2+m$ constraints specifying the ranges of values the variable can take. Consequently, the number of constraints is $O(n^2,mn^2)$.

\citet{lenstra1983integer} had shown that for every ILP, there is a corresponding parameterized algorithm that solves the ILP, when parameterized by the number of integer variables. As a result, our ILP leads to the following result.

\begin{proposition}\label{thm:pess-fpt}
    There is an $O^*(2^{n(n-1+m)})$ time algorithm for finding a minimum sized pessimistically valid menu. 
\end{proposition}
\subsection*{Parameterized Algorithm for Optimistic Menus}

In this section we design a parameterized algorithm to find a minimum size menu that is optimistically valid. Towards this first we define maximum weighted partition. A tuple of subsets $(S_1, S_2, \dots S_k)$ of a set $N$ is a partition of the set $N$ if
$\cup_{i \in [k]} S_i = N$ and $S_i \cap S_j =\emptyset$ for each pair $i,j \in [k]$, and $i \neq j$.
The \textsc{Maximum Weighted Partition} problem is defined as follows \cite{bjorklund2009set}:
Given a universe $N$, a set family $\mathcal{F}$ over $N$, a non-negative integer $k$, and functions $f_1, f_2, \dots f_k$, we ask if there is a partition $(S_1, S_2, \dots S_k)$ of $N$ in to $k$ parts such that  the sum $\sum_{i \in [k]} f_i(S_i)$ is maximum over all partitions $(S_1, S_2, \dots S_k)$ of $N$, where $S_1, S_2, \dots S_k$ are from $\mathcal{F}$.
\\

\noindent
\textbf{Algorithm Overview.} Let $k$ denote the number of options that will be used in the menu. We iterate over the values of $k$ starting from $1$ to $m$. For a given value of $k$, we reduce our problem to a maximum weighted partition problem. We define the universe as the set of agents $N$. Let the set of options be $\{o_1,o_2, \dots, o_m\}$ with serving sizes $s_1, s_2, \dots, s_m$ respectively. For each option $o_j \in O$, we define $N_j$ to be the set of agents that accepts $o_j$.
We define a family $\mathcal{F}$ of subsets of $N$ as follows: $\mathcal{F}$ contains all the non-empty subsets of $N_j$ for each $o_j \in O$. We define a function $f: \mathcal{F} \rightarrow \mathbb{Z}$ as follows.
For a set  $N' \subseteq N$ in the family $\mathcal{F}$, let $O'$ denote the set of options that are accepted by all agents in $N'$ and $o' \in argmin_{j' \in O'} \{|N'| = r' \mod s_{j'}\}$. Then,  we define $f(N') = -r$ where $|N'| = r \mod s_{o'}$. 
Now we find a partition $(X_1,X_2, \dots X_k)$ of $N$ into $k$ parts such that $X_i \in \mathcal{F}$ for each $i \in [k]$, and $\sum_{i \in [k]} f(X_i)$ is maximum. Let $s_{max}$ denote the maximum serving size of any option, i.e., $s_{max} =\max_{j \in O} s_j$.  We iterate over the values of $k$ starting with $1$ and going up to $m$ and take the solution partition $(X_1,X_2, \dots X_{k^*})$ of $N$   that has maximum weight over all values of $k$. For a part $X_z$ where $z \in k^*$, let $O_z$ denote the set of options that are accepted by all agents in $X_z$ and $o_z \in argmin_{j' \in O_z} \{|N'| = r' \mod s_{j'}\}$. Then we construct $M$ by setting $q_{o_z}(M) = \Big\lceil\frac{X_z}{s_{o_z}}\Big\rceil$ for each $z \in k^*$; for any other option $o\in O$, we set $q_o(M)=0$.

The maximum weighted partition problem can be solved in time $O^*(2^nk^2s_{max})$, as shown by \citet{bjorklund2009set}. We run the algorithm for maximum weighted partition at most $m$ times. Therefore, our final running time is $O^*(2^nm^3s_{max})$.

\optfpt*

\noindent
\textbf{Correctness of the Algorithm.}
We now prove the correctness of our algorithm and analyze its running time.
Recall that given an instance 
$I=\langle N,O,(A_i)_{i\in N},(s_j)_{j\in O}\rangle$, 
our goal is to select a optimistically valid menu $M$ minimizing its \emph{size}, where
\[
size(M)\;=\;\sum_{j\in O} q_j(M)\cdot s_j,
\]
where $q_j(M)$ denotes the multiplicity of option $o_j$ in $M$.
A menu $M$ is \emph{optimistically valid} if $\lambda(M)=n$, i.e.\ every agent
can be served fractionally with acceptable options.

To show the correctness of the algorithm, we show that the weight a  partition of $N$ (as defined in the algorithm) is the same as the waste of the menu that is constructed using the last step of our algorithm. 

Recall that for each option $o_j$, let $N_j=\{\,i\in N: o_j\in A_i\,\}$. 
We define a set family $\mathcal F=\{X\subseteq N_j| X \neq \emptyset \text{ and } o_j\in O\}$.
For each $X\in\mathcal F$, define its \emph{waste} as
\[
waste(X) := \min_{j: X\subseteq N_j} 
\left((s_j-(|X|\bmod s_j))\bmod s_j\right).
\]
That is, the minimum waste that created when all agents in $X$ feeds from the same option. recall that we define the weight function $f(X)=-waste(X)$.

For a partition $\mathcal X=(X_1,\dots,X_k)$ of $N$, the algorithm
maximizes
\[
F(\mathcal X) = \sum_{i=1}^k f(X_i) 
= -\sum_{i=1}^k \mathrm{waste}(X_i).
\]
Therefore, the algorithm minimizes the waste when the agents in $X_z$ feed from the same option for each $z \in [k]$. Note that there are $m$ options, then there can be at most $m$ sets of agents who feed from the same option. Since our algorithm take the partition that has minimum sum over all values of $k$, we find the minimum waste.

The following two lemmas show the equivalence between a valid menu and a maximum size partition of each value of $k$, then finally, use these lemmas we finally complete the proof of correctness.

\begin{lemma}[Partition $\Rightarrow$ Menu]\label{lem:partitiontomenu}
Let $\mathcal X=(X_1,\dots,X_k)$ be a partition of $N$. Let menu $M$ is constructed as follows:
for each part $X_z$, pick an option $o_z$ that attains the minimum in 
${waste}(X_i)$ and set 
$q_{o_z}=\lceil |X_i|/s_{o_z}\rceil$.  
Then
\[
\text{size}(M) \;=\; n + \sum_{i=1}^k \mathrm{waste}(X_i) 
\;=\; n - F(\mathcal X).
\]
\end{lemma}

\begin{proof}
Since $X_i\subseteq N_{j(i)}$, all agents in $X_i$ accept option $j(i)$.
The serving capacity for part $X_i$ is $q_{j(i)}s_{j(i)}$, which equals
$|X_i|+\mathrm{waste}(X_i)$. Summing over the parts yields
\[
\sum_j q_j(M) s_j = \sum_{i=1}^k \big(|X_i|+\mathrm{waste}(X_i)\big) 
= n + \sum_{i=1}^k \mathrm{waste}(X_i).
\]
By definition of $f$, this equals $n - F(\mathcal X)$. 
\end{proof}

\begin{lemma}[Menu $\Rightarrow$ Partition]\label{lem:menutopartition}
Let $M$ be any menu with at most $k$ distinct options. 
There exists a partition $\mathcal X=(X_1,\dots,X_k)$ of $N$ such that
\[
\text{size}(M) \;=\; n - F(\mathcal X).
\]
\end{lemma}

\begin{proof}
Let the distinct options in $M$ be $j_1,\dots,j_{k'}$ ($k'\le k$). 
Consider an optimal fractional allocation witnessing $\omega(M)=n$.
For each $j_t$, let $X^{(t)}$ be the set of agents served by option $j_t$.
These $X^{(t)}$ form a partition of $N$ (pad with empty parts if needed).
The waste for part $X^{(t)}$ equals $q_{j_t}(M)s_{j_t}-|X^{(t)}|$.
Choosing option $j_t$ in the definition of $f(X^{(t)})$ gives
$f(X^{(t)})=-\mathrm{waste}(X^{(t)})$. Summing over all parts yields
$F(\mathcal X)=-\sum_t \mathrm{waste}(X^{(t)})$, hence
$\text{size}(M)=n - F(\mathcal X)$.
\end{proof}
Using \Cref{lem:partitiontomenu,lem:menutopartition}, we obtain the following proposition.
\begin{proposition}[Correctness]\label{thm:parametrized}
For each $k$, let
\[
\mathrm{opt}_k := \max_{\mathcal X\text{ partition into }k\text{ parts}} F(\mathcal X).
\]
Then
\[
\min_{M:\,M\text{ valid}} \text{size}(M) \;=\; 
\min_{k=1}^m \big(n - \mathrm{opt}_k\big).
\]
That is, the algorithm computes $\mathrm{opt}_k$ for all $k=1,\dots,m$,
and returns the menu induced by the partition achieving the maximum
objective value. This menu is optimistically valid and has minimum possible size.
\end{proposition}
\begin{proof}
    We show that the size of the menu $M$ constructed using a maximum weight partition $\mathcal X$ is minimum. Since $\mathcal X$ is maximum weight partition, $F(\mathcal X)$ is maximized over all partitions and  $\sum_{i=1}^k {waste}(X_i)= F(\mathcal X)$, we have that the total waste is minimized. As any feasible menu must have size $n$ and $M$ has size $n - F(\mathcal X)$,  the size of $M$ is minimum.
\end{proof}
\paragraph{Running Time Analysis}

For each fixed $k$, we must solve the maximum weighted partition problem
with universe $N$, family $\mathcal F$, and weight function $f$.
By the result of Björklund, Husfeldt, and Koivisto
\cite{bjorklund2009set}, this can be done in time
$O^*(2^n k^2 s_{\max})$,
where $s_{\max}=\max_{j\in[m]} s_j$ and the $O^*$ notation suppresses polynomial factors in $n$. 
Since the algorithm must run this routine for each $k=1,2,\dots,m$, 
the overall running time is
$O^*(2^n m^3 s_{\max})$.

\section{Omitted Material from \cref{sec:algorithms}}\label{sec:exactalgos}

\subsection*{Limited Serving Sizes}

 
\smallplates*
\begin{proof}
    Construct the following auxiliary weighted graph $G=(V,E)$ where there is one vertex for each agent $i\in N$ and for each option $o$ and agent $i\in N$ s.t. $o\in A_i$ and $s(o)=1$, there is a vertex $o^i$. Here, add an edge between $i,j\in N$ only if there exists $o\in A_i\cap A_j$ s.t. $s(o)=2$ and set $w(i,j)=2$. Further, add an edge between $i\in N$ and a vertex $o^i$ for $o\in O$ where $s(o)=1$ and $o\in A_i$ and set $w(i,o^i)=1$. 

    Consider an arbitrary matching $\mu$ in $G$. Observe that each edge $e\in \mu$ must contain at least one agent. For each edge $e=(i,j)$, we can add option $o$ s.t. $o\in A_i\cap A_j$ and $s(o)=2$. For each edge $(i,o)\in \mu$, it must be that $o\in A_i$ and $s(o)=1$. Thus, based on a matching, we can prepare a menu that feeds all matched agents in $\mu$. 

    We shall now show that a maximum weight matching in $G$ has weight $n-t$ for $t\in [n]$ if and only if a minimum size optimistically valid menu has size $n+t$. 
    
    First consider a minimum size optimistically valid menu $M$. As size of each option is at most $2$, $s(M)\leq 2n$. Further, if $s(M)=n+t$, exactly $t$ distinct options in $M$ of size $2$ must have been ordered that are only consumed by one agent. For each such agent $i$, observe that $A_i$ must only contain options of size $2$. Further, all other agents who find an option $A_i$ must already be sharing a different option of size $2$ with another agent $i'$. It is easy to see that $A_{i'}$ must also only consist of options of size $2$. All  options in $M$ that are fully consumed, are either of size $1$ or each ordered quantity is consumed by two agents. 
    Consequently, we can create a matching in $G$ based on $\mu$ of weight $n-t$.
    
    Now, consider a maximum weight matching. Let $M_\mu$ be the corresponding menu. $M$ is not optimistically valid if and only if there are some unmatched agents under $\mu$. As $\mu$ is a maximum weight matching, it must be that for each unmatched $i$, $A_i$ only contains options of size $2$. Further, all agents who also like an option in $A_i$ must be matched to another agent. If they were matched to an option, being matched to $i$ would increase the weight of the matching. Thus, it in unavoidable to order some $o\in A_i$ just for $i$.  As a result, for each unmatched $i$, we can arbitrarily choose an option $o$ from $A_i$ and add it to $M_{\mu}$. Thus, if $w(\mu)=n-t$, exactly $t$ options of size $2$ are ordered for the $t$ unmatched agents, making $M_{\mu}$ have size $n-t+2t=n+t$.

    \paragraph{Running Time.} The graph constructed consists of at most $n+nm$ vertices and at most $n^2+nm$ edges. A maximum weight matching can be found using the blossom algorithm \cite{edmonds1965paths} in time $O(n^3(nm^2+m^3))$. A maximum weight matching can have weight at most $n$ and from it a menu can be constructed in time $O(n+m)$. As a result, the algorithm runs in polynomial time. 
\end{proof}

\subsection*{Structured Acceptability Relations}\label{app:laminar}



\laminarstructure*

\begin{proof}
    Let $I$ be laminar. Fix $i,j\in N$ s.t. $|A_i|\leq |A_j|$. If $A_i\subseteq A_j$ then $A_i\setminus A_j$ must be non-empty. Else if $A_i$ and $A_j$ are disjoint, clearly, $A_i\setminus A_j=\emptyset$. 

    For the converse, let $I$ be such that  for every $i,j\in N$ at most two of $A_i\cap A_j$, $A_i\setminus A_j$ and $A_j\setminus A_i$ are non-empty. As a result, at least one of these sets must be empty. Consider any two $i,j\in N$. If $A_i\cap A_j$ is empty, clearly, they satisfy laminarity. If $A_i\setminus A_j$ is empty,it must hold that $A_i\subseteq A_j$. Similarly, if $A_j\setminus A_i$ is empty, it must hold that $A_j\subseteq A_i$. Therefore, $I$ satisfied laminarity. 
\end{proof}

\paragraph{Identical Acceptability Relations} We first look at the case when all agents find the same set of options acceptable.  As agents are identical, we can discard options that are not liked by any agent. Consequently, we can assume, without loss of generality, that  $O=A_i$ for each $i\in N$. Consequently, every menu of size at least $n$ will be both optimistically and pessimistically valid.  
Options with serving size $n$ or more need only one quantity ordered, without ordering anything else, but options with smaller serving sizes will have to be combined to feed all agents. So we choose either the smallest sized "large option" or the smallest possible combination of "small options". For small options, we effectively run a subset sum type DP to find the smallest menu that can feed everyone.


\begin{algorithm}[t]
    \DontPrintSemicolon
  \KwIn{ \Instance{} instance $ \langle N, A, (s_o)_{o\in O} \rangle$ with identical acceptability relations}
  \KwOut{Minimum size of valid menu $c$}

        Initialize small options set $A^s\gets \{o\in A|s_o< n\}$\;
        Min sized large options set $A^\ell\gets \argmin \{s_o|o\in A, s_o\geq n\}$\;
        \leIf{$A^\ell\neq \emptyset$}{
            $c^\ell\gets s_o$ where $o\in A^\ell$\;
            }{
                $c^\ell\gets \infty$
            }
        $c^s\gets \infty$\;
        \If{$A^s\neq \emptyset$}{
            $\ess\gets \{ ks_o|o\in A^s,\,k\leq \lceil \frac{n}{s_o}\rceil\}$\;
            
            Let $s_1 \leq \cdots\leq  s_{|\mathcal{S}|}$ be the elements of $\ess$\;
            $c^s\gets \textsc{Ident}(\ess,n,2mn) $\;
        }
        
        \textbf{Return} $\min c^s,c^l$\;
    
   \caption{Minimum Sized Valid Menus under Identical {Acceptability Relations}} \label{alg:IdentOptMenu}
\end{algorithm}

\begin{algorithm}[t]
    \DontPrintSemicolon
  \KwIn{$\langle \ess, goal, t\rangle $}
  \KwOut{Cost $c$}
        Let $s_1 \leq \cdots\leq  s_{|\mathcal{S}|}$ be the elements of $\ess$\;
        $t\gets \min (|\ess|,t)$\;
        \If{$t>1$}{
            \eIf{$s_t\geq goal$}{
                $t_{in}\gets s_t$\;
            }{
                $t_{in}\gets s_t+ \textsc{Ident}(\ess,goal-s_t,t-1)$
            }
            $t_{out}\gets \textsc{Ident}(\ess,goal,t-1)$\;
            \textbf{Return} $\min (t_{in},t_{out})$\;
        }
        \Comment{$t=1$}\;
        \If{$s_1\geq goal$}{
            \textbf{Return} $s_1$\;
        }
        \textbf{Return} $\infty$\;
    
   \caption{$\textsc{Ident}$} \label{alg:IdentDP}
\end{algorithm}

\begin{proposition}\label{thm:identical}
        Given identical \Instance{} instance $I=\langle N,A,(s_o)_{o\in A}\rangle$, \cref{alg:IdentOptMenu} finds the minimum size of an optimistically valid menu in polynomial time. This will also be the minimum size of a pessimistically valid menu.
\end{proposition}

\begin{proof}

    When agents have identical acceptability, we use $A\subseteq O$ to denote the universally accepted set of options, that is, for any $i\in N$, $A_i=A$. 

    We first show that every optimistically valid menu will also be pessimistically valid. Consider an optimistically valid menu $M$ and the graph where there is a vertex for each agent and each option, with an edge from $i$ to $o$ if $o\in A_i$. Set vertex capacities where $c(i)=1$ for each $i\in N$ and $c(o)=q_o(M)$ for each $o\in O$.

    Now as $M$ is optimistically valid, there must be a matching $\mu:N\times O\rightarrow \mathbb{R}^+$ s.t. $\sum_{o\in A_i} \mu(i,o)=1$ and $\sum_{i\in N} \mu(i,o)\leq q_o(M)s_o$. In order to ensure that $M$ is pessimistically valid, we need that for any other maximal matching $\mu'$, we have that $\sum_{o\in A_i} \mu'(i,o)=1$. Recall that $\mu'$ is a maximal matching if for every $i\in N$ s.t. $\sum_{o\in O}\mu'(i,o)<1$, we have that for each $o\in A_i$, $\sum_{i\in N} \mu'(i,o)=q_o(M)s_o$. As agents have identical acceptability sets, any maximal matching $\mu'$ must in fact satisfy $\sum_{o\in O} \mu'(i,o)=1$ for all $i\in N$. Consequently, $M$ must be pessimistically valid.

    It now remains to construct a minimum sized valid menu. Specifically, it suffices to find the smallest sized menu $M$ s.t. $\sum_{o\in A}q_o(M)\geq n$. Observe that if a menu contains even one option $o\in A$ s.t. $s_o\geq n$, even if we discard all other options from this menu, it would continue to be valid. Consequently, a minimum sized valid menu would only involve multiple options if all of them have serving size less than $n$. We use this approach in \Cref{alg:IdentOptMenu} and its subroutine \Cref{alg:IdentDP}. 

    The algorithm proceeds by separating options into a small options set $A^s$ consisting of the options of size less than $n$ and picks the smallest size large options in the set $A^{\ell}$ containing the smallest options of size at least $n$, if any. Only one these options is needed to make a valid menu.  For the small options set, \cref{alg:IdentOptMenu} creates a set $\mathcal{S}$, consisting of one element corresponding to $k$ copies of $o\in A^s$ where $k\leq \lceil \frac{n}{s_o}\rceil$. 
    
    At this point the dynamic program $\textsc{Ident}$ (\cref{alg:IdentDP}) is invoked with $\ess$ and the goal of $n$. Starting from the largest sized element in $\ess$, $\textsc{Ident}$ proceeds by considering two cases for each element $s_t$: in or out, captured by $t_{in}$ and $t_{out}$. For $t_{in}$, the goal is reduced $s_t$ and the remaining is filled by $s_{t-1}$ or smaller elements. For $t_{out}$, $s_t$ is not a part of the solution and the goal must be filled by $s_{t-1}$ or smaller elements. The minimum of these two options is returned. 

    \paragraph{Running Time.} Observe that the running time of \cref{alg:IdentOptMenu} depends largely on the running time of \cref{alg:IdentDP} (\textsc{Ident}), which is a dynamic program. The size of the dynamic program is based on the size of the set $\ess$ which contains $\lceil \frac{n}{s_o}\rceil$ elements for each $o\in O$. Consequently, the size of $\ess$ is at most $nm$, and consequently, $\textsc{Ident}$ has size at most $O(n^3m^2)$. 
\end{proof}

\subsection*{Optimistically Valid Menus for Laminar Acceptability Relations.}    

Under laminar acceptability relations, the acceptability sets of two agents are either disjoint, or one contains the other. Laminar sets can be represented by a containment tree where if $A^i\subsetneq A^j$, then $A^j$ must be an ancestor of $A^i$ in the containment tree. Let $A^1,\cdots, A^k$ be the distinct acceptability sets labeled in a pre-order traversal of the containment tree. That is, if $A_i\subseteq A_j$ then $i\geq j$. The agents with acceptability set $A^i$ can be fed from a combination of items across $A^i$ or a combination of one child alone. Consequently, to find min sized menus, we need to determine the best combination of agents being fed from each acceptability set. To this end, for each internal node $A^i$, we need to decide how many of the corresponding agents eat from $A^i$ and how many eat from each child of $A^i$. 

To avoid going over all possible combinations here, we use the tree structure to build these combinations sequentially. Each node $v$ other than the root inherits a ``debt" from either its previous sibling (according to the label) or its parent (if no sibling comes earlier in the labeling), but not both. Now $v$ needs to decide how much of this debt it will keep for the subtree rooted at $v$, and how much to pass of onto the next sibling, if any. Now the debt it retains plus $n_v$ is the number of agents that need to be fed from the subtree rooted at $v$. Of this number, $v$ now decides how many to feed from the set $A^v$ and passes off the remaining to the first child (according to the labelling), if any.

\begin{algorithm}[t]
    \DontPrintSemicolon
  \KwIn{ Menu Planning instance $ \langle N, O, (A_i)_{i\in N}, (s_o)_{o\in O} \rangle$ with laminar acceptability relations}
  \KwOut{Minimum size of optimistically valid menu $c$}
    
        Let $A^1,\cdots A^k$ be the distinct acceptability sets s.t. sets are labeled in a pre-order traversal of the associated containment tree where $\{1,2, \dots ,k\}$ denote the acceptability sets ordered according to the pre-order traversal\;
        Let $n_v\gets |\{i\in N|A_i=A^v\}|$\;
        Let $sib(v)\gets \min \{v'|v' > v$ AND $v'$ is a sibling of $v\}$ if any, $0$ otherwise\;
        Let $child(v)\gets \min \{v'|v'$ is a child of $v\}$ if any, $0$ otherwise \;
        Initialize $k\times n$ array $LaminarOpt$\;
        \For{$v=k$ to $1$}{
        Let $t = \sum_{v':A^v\subsetneq A^{v'}} n_{v'}$ \;
            \For{$debt= t$ to $0$}{
                $u_1\gets \mathds{1}(sib(v) > v)debt$\;
                $u_2\gets \mathds{1}(child(v)> v)(n_v+debt)$\;
                $\ess^v\gets \{ ks_o|o\in A^v,\,k\leq \lceil \frac{n}{s_o}\rceil\}$\;
                $LaminarOpt[v][debt]\gets \min_{0\leq t_1\leq u_1}\min_{0\leq t_2\leq \max (u_2-t_1,0)} \Big( \textsc{Ident}(\ess^v,n_v+debt-t_1-t_2,nm) $\\
                \quad \quad \quad \quad $+\mathds{1}(sib(v)>v)LaminarOpt[sib(v)][t_1]+\mathds{1}(child(v)>v)LaminarOpt[child(v)][t_2]\Big)$\;
            }
        }        
        \textbf{Return} $LaminarOpt[1][0]$\;
   \caption{Minimum Sized Optimistically Valid Menus under Laminar {Acceptability Relations}} \label{alg:LaminarOptMenu}
\end{algorithm}

\paragraph{Algorithm Overview.} \cref{alg:LaminarOptMenu} starts at the root of the laminar containment tree, $A^1$ that is the largest acceptability set. Recall that we assume that the vertices of the tree are labeled in a pre-order traversal, so each internal vertex is labeled before its children. We say that $n_t$ agents have acceptability set $A^t$. We use this ordering to run the dynamic program to sequentially decide, how many agents must eat from each vertex of the tree and use the algorithm for identical acceptability relations to find a minimum sized menu for these many agents. 

    We introduce the following notation: $sib(v)$ is used to denote the next sibling of $v$ according to the labels, if any. If $v$ has no siblings with a higher label, $sib(v)$ is set to $0$. Similarly, $child(v)$ is used to denote the child of $v$ with the smallest label. If $v$ has no children, the value of $child(v)$ is set to $0$. Further, we use $debt$ to denote the number of agents whose acceptability set is a superset, but must eat from the subtree rooted at $v$ or its siblings. 

    Each vertex $v\in [k]$, other than the root, receives a $debt\geq 0$. If $v$ has a sibling, some of this debt can be passed on to them. We set $u_1$ as equal to $debt$ if $sib(v)>0$. That is, $u_1$ is set to the maximum debt that can be sent to $sib(v)$. We choose $t_1$ between $0$ and $u_1$ as the debt passed onto $sib(v)$. The remaining $debt-t_1$ agents must be fed from the subtree rooted at $v$. We now can pass some off these agents on to the first child $child(v)$, if any. We analogously set $u_2$ as $n_v+debt-t_1$ if $child(v)>0$ and choose $t_2$ as the debt passed on to $child(v)$. We then use \cref{alg:IdentOptMenu} to find the minimum sized menu to feed the remaining $n_v+debt-t_1-t_2$ agents from $A^v$. 

    We make these decisions by using a dynamic program that we call $LaminarOpt[v][debt]$ where $v\in [k]$ and $debt$ can take any integral value between $0$ and $\sum_{v':A_v\subsetneq A_{v'}}n_{v'}\leq n$. We fill this from $v=k,\cdots, 1$ and for all values of $debt\in \{0,\cdots, \sum_{v':A_v\subsetneq A_{v'}}n_{v'}\}$. Observe that $v$ can only pass $t_1>0$ debt to $sib(v)$ if $sib(v)>v>0$. Similarly, as the vertices are labeled in a pre-order traversal of the tree, each child's entries would be filled before we reach their parent. As a result, we can fill each entry of $LaminarOpt[v][debt]$ by considering at most $n^2$ other entries. We ultimately return $LaminarOpt[1][0]$.

\laminaropt*

\begin{proof}
    Fix a \Instance{} instance $I=\langle N, O, (s_o)_{o\in O},(A_i)_{i\in N}\rangle$ with laminar acceptability relations. Let $A^1,\cdots, A^k$ be the distinct acceptability sets across all agents in $N$. Further, assume that these sets are labeled s.t. if $A^{j}\subseteq A^{j'}$, we have that $j\geq j'$. As these sets are laminar, we assume that they are labeled in a {\em pre-order traversal} of the associated containment tree.\footnote{A pre-order traversal of tree is one where each parent vertex is visited before its corresponding children. Given a collection of laminar sets, we can construct a corresponding containment tree s.t. if $S\subsetneq T$ then the vertex corresponding to $T$ must be an ancestor of the vertex corresponding to $S$.} 

    Consider an agent $i\in N$. Consider an optimistically valid menu $M$ and a maximum sized integral matching $x$. As $M$ is optimistically valid, there must exist $o\in A_i$ s.t. $x_{i,o}=1$. Let $A^j$ be the smallest acceptability set that contains $o$. That is, $A^j=\argmax \{t\in [k]|o\in A_t\}$. Observe that $A^j$ must be the same as $A_i$ or a descendant of it in the laminar containment tree. In other words, $i$ must eat from $A_i$ or a descendant of it in the containment tree. 
    
    Consequently, in order to find min sized optimistic menus, we need to decide how many agents eat from each set. We can find min sized menus for each set using the identical case (\cref{alg:IdentOptMenu}) as a subroutine. However, there may be exponentially many ways to divide agents among these sets. We present a clever way of making this decision sequentially and using dynamic programming to combine different independent subcases. 

    For each $t\in [k]$ net $n_t$ be the number of agents who have acceptability set $A^t$. The agents eating from $A_t$ can be agents with this exact acceptability set or with a superset (an ancestor of $A_t$). These agents can also eat from a descendant of $A_t$, and those with strictly larger acceptability sets can also eat from a sibling of $A_t$. Consequently, \cref{alg:LaminarOptMenu} uses a dynamic program in a top down manner on the laminar containment tree as follows.

    \paragraph{Correctness of Algorithm.} We shall first show that {\em every} optimistically valid menu $M$ can be broken into a distribution of agents eating from each set $A^t$ for $t\in [k]$ and thus be would be considered by \cref{alg:LaminarOptMenu}. 
    
    Fix an arbitrary optimistically valid menu $M$. Choose a maximum cardinality optimistic consumption matching $x$ s.t. $\sum_{i\in N,o\in O} x_{i,o}=n$. We shall now count the number of agents eating from each set $A^v$ using the variables $p_v$. Initialize $p_v$ to $0$ for each variable $v\in [k]$.
    
    For each option $o\in \widehat{M}$ s.t. $\sum_{i}x_{i,0}>0$, let $A^v$ be the smallest acceptability set containing $o$. That is, for each $v'\neq v$ s.t. $o\in A^{v'}$, $A_v\subsetneq A^{v'}$. We now increase $p_v$ by $\sum_i x_{i,o}$ to account for the number of people eating $o$. Now, having done this for all distinct options in $M$, we now must have that $\sum_{v\in [k]} p_v=n$, as $M$ is optimistically valid. 
    
    Further, for each $v\in [k]$, we only count the agents eating from it as those that are eating an option not contained in one of its descendants. Thus, the acceptability set of all these agents must be $A^v$ or a superset of $A^v$. Consequently, it must hold that $p_v\leq n_v + \sum_{v':A_v\subsetneq A_{v'}}n_{v'}$. Additionally, for each leaf $v$, it must be that $p_v\geq n_v$ as agents with acceptability set $A^v$ cannot eat from any other set. As a result, $p_v-n_v$ is precisely the number of agents with strictly larger acceptability sets. Thus, starting from the last leaf, we can calculate the $debt$ inherited by each leaf in the tree. This then helps us calculate the $debt$ inherited by each vertex whose children are all leaves. In this manner, we can calculate the $debt$ inherited by each vertex and consequently the $t_1$ and $t_2$ values of each vertex. As we solve each subcase optimally using \cref{alg:IdentOptMenu}, for this sequence of $t_1$ and $t_2$ values for each vertex, the menu size stored would be less than or equal to that of $M$ but still feed the same number of agents from each set $A^v$. 

    As a result, we have that each valid menu will be considered or something that feeds the same number of people with a smaller size. It is straightforward to see that the algorithm does not consider any invalid division of agents across the sets. Hence, as the algorithm returns the minimum size encountered, it must find a minimum sized optimistically valid menu for $I.$

    \paragraph{Running Time.} Observe that the size of the $LaminarOpt$ is at most $n^2$. The laminar tree can contain at most $n$ vertices, and the debt to each agent can be at most $n$. To fill each entry, \cref{alg:IdentOptMenu} is called at most $n^2$ times and combined with other existing entries of $LaminarOpt$. Further \cref{alg:IdentOptMenu} runs in time $O(n^3m^2)$. Consequently, $LaminarOpt$ runs in time $O(n^7m^2)$. 
\end{proof}

\begin{algorithm}[t]
    \DontPrintSemicolon
  \KwIn{ Menu Planning instance $ \langle N, O, (A_i)_{i\in N}, (s_o)_{o\in O} \rangle$ with laminar acceptability relations}
  \KwOut{Minimum size of valid menu $c$}
    
        Let $A^1,\cdots A^k$ be the distinct acceptability sets\;
        Let $n_v\gets |\{i\in N|A_i=A^v\}|$\;
        Consider the set of leaf vertices $\mathcal{L}\gets \{v\in [k]|v$ is a leaf $\}$\;
        Initialize $c\gets 0$\;
        \For{$v\in \mathcal{L}$}{
            Let $t = n_v + \sum_{v':A^v\subsetneq A^{v'}} n_{v'}$ \;
            $\ess^v\gets \{ ks_o|o\in A^v,\,k\leq \lceil \frac{n}{s_o}\rceil\}$\;
            $c\gets c+ \textsc{Ident}(\ess^v,t,nm)$\;
        }        
        \textbf{Return} $c$\;
   \caption{Minimum Sized Pessimistically Valid Menus under Laminar {Acceptability Relations}} \label{alg:LaminarPessMenu}
\end{algorithm}

\subsection*{Pessimistically Valid Menus for Laminar Acceptability Relations.}
We now build pessimistically valid menus for the laminar setting. To this end, we first give a necessary and sufficient condition for pessimistic validity under laminarity.

\begin{lemma}\label{lem:LaminarPess}
    Given a \Instance{} instance $I$ with laminar acceptability relations and menu $M$. Let $v_1,\cdots,v_t$ be the leaf nodes of the laminar containment tree. $M$ is pessimistically valid if and only if  for any $v\in \{v_1,\cdots,v_t\}$, we have that $\sum_{o\in A^{v}} q_M(o)s_o\geq  n_{v} + \sum_{v':A^{vj}\subsetneq A^{v'}} n_{v'}$. 
\end{lemma}

\begin{proof}
    Fix a \Instance{} instance $I=\langle N, O, (s_o)_{o\in O},(A_i)_{i\in N}\rangle$ with laminar acceptability relations and a menu $M$. Let $A^{v_1},\cdots, A^{v_t}$ be the $t$ distinct minimal acceptability sets. That is, these sets are all pairwise disjoint and for all agents $i\in N$, there must exist at least one $j\in [t]$ s.t. $A^{v_j}\subseteq A_i$. It is straightforward to see that these must be the leaves of the laminar containment tree.

    Let $M$  be a pessimistically valid menu. Recall from \cref{lem:pess-valid} that for a menu $M$ to be pessimistically valid, for each agent $i\in N$ it must hold that $size(M|_{A_i})=\sum_{o\in A_i}q_o(M)s_o\geq 1+\rho_i(M)$, where $M$ is the maximum cardinality matching from $N\setminus \{i\}$ to $M|_{A_i}$. Consider an agent $i\in N$ s.t. $A_i=A^v$ for some $v\in \{v_1,\cdots, v_t\}$. 
    
    Consider the graph that determines $\rho_i$. One side of the bipartition will have $M|_{A_i}$ and the other has the set $N^i=\{j\in N\setminus \{i\}|A_i\cap A_j\neq \emptyset\}$. Observe that due to laminarity and the fact that $A_i=A^v$ is a minimal acceptability set, every agent $j\in N\setminus\{i\}$ s.t. $A_i\cap A_j\neq \emptyset$ satisfies $A_i\subseteq A_j$. As a result, every such agent $j$ will have an edge to every option in $A_i$. Consequently,  $\rho_i(M)=\min (n_{v}-1 + \sum_{v':A^{v}\subsetneq A^{v'}} n_{v'},size(M|_{A_i}))$. As $M$ is pessimistically valid, it must hold that $\sum_{o\in A^v} q_o(M)s_o\geq n_{v} + \sum_{v':A^{vj}\subsetneq A^{v'}} n_{v'}$.


    

    %
    The remaining case is to show that when $\sum_{o\in A^{v}} q_M(o)s_o\geq  n_{v} + \sum_{v':A^{v}\subsetneq A^{v'}} n_{v'}$ holds for each $v\in \{v_1,\cdots,v_t\}$, $M$ must be pessimistically valid. Let $\sum_{o\in A^{v}} q_M(o)s_o\geq  n_{v} + \sum_{v':A^{v}\subsetneq A^{v'}} n_{v'}$ hold for each $v\in \{v_1,\cdots,v_t\}$. Consider an arbitrary agent $i\in N$ and minimal acceptability sets $v_{i_1},\cdots,v_{i_t}$ s.t. for each $v\in \{v_{i_1},\cdots, v_{i_t}\}$, we have that $A^v\subseteq A_i$. As $I$ has laminar acceptability relations, we have that for every agent $j\in N\setminus \{i\}$ s.t. $A_i\cap A_j\neq \emptyset$, either $A_i\subseteq A_j$ or $A_j\subsetneq A_i$. Let $N^i_+=\{j\in N\setminus \{i\}|A_i\subseteq A_j\}$ and let $N^i_-=\{j\in N\setminus\{i\}|A_j\subsetneq A_i\}$. 
    
    Observe that for each $j\in N^i_+$, we have that $A^v\subseteq A_j$ for every $v\in\{v_{i_1},\cdots,v_{i_t}\}$. On the other hand, for each $j\in N^i_-$, we have that there must exist $v\in\{v_{i_1},\cdots,v_{i_t}\}$ s,t, $A^v\subseteq A_j$. As a result, we have that the number of agents who find an option in $A_i$ acceptable is exactly $|N^i_+|+|N^i_-|=(\sum_{v':A^{v'}\cap A_i\neq \emptyset}n_v) -1$. 

    Recall that for every $v\in\{v_{i_1},\cdots,v_{i_t}\}$, we have that $\sum_{o\in A^{v}} q_M(o)s_o\geq  n_{v} + \sum_{v':A^{v}\subsetneq A^{v'}} n_{v'}$. Consequently,  for every $v\in\{v_{i_1},\cdots,v_{i_t}\}$, we have that $\sum_{o\in A^{v}} q_M(o)s_o\geq 1+ |N^i_+|+|\{j\in N^i_-|A^v\subseteq A_j\}|$. Thus, $M$ is pessimistically valid.   


\end{proof}

\laminarpess*

\begin{proof}
    Fix a \Instance{} instance $I=\langle N, O, (s_o)_{o\in O},(A_i)_{i\in N}\rangle$ with laminar acceptability relations. We shall prove that \cref{alg:LaminarPessMenu} returns the minimum possible size of a pessimistically valid menu under $I$. We first begin with an algorithm overview.

    \paragraph{Algorithm Overview.} \cref{alg:LaminarPessMenu} proceeds by considering the laminar containment tree and the $t$ leaves $A^{v_1},\cdots, A^{v_t}$. For each leaf $v_j$ it adds the minimum size pessimistic menu for feeding $n_{v_j} + \sum_{v:A^{v_j}\subsetneq A^{v}} n_{v}$ from $A^{v_j}$ using the algorithm for identical acceptability relations as a subroutine. The combined sizes for each leaf is returned.

    \paragraph{Correctness.} Firstly, it is straightforward to see that the menu constructed by \cref{alg:LaminarPessMenu} would satisfy \cref{lem:LaminarPess}, thus it is clearly pessimistically valid. 
    
    Let $M'$ be a minimum sized pessimistically valid menu. From \cref{lem:LaminarPess}, we have that for each minimal acceptability set $A^{v_j}$, $\sum_{o\in A^{v_j}} q_{M'}(o)s_o \geq n_{v_j} + \sum_{v:A^{v_j}\subsetneq A^{v}} n_{v}$. Now as \cref{alg:IdentOptMenu} always returns the minimum menu size to feed a given number of (identical) agents from a given set, it must be that $\sum_{o\in A^{v_j}} q_{M'}(o)s_o \geq \textsc{Ident}(A^{v_j},)n_{v_j} + \sum_{v:A^{v_j}\subsetneq A^{v}} n_{v}$. If this is not so, either \cref{alg:IdentOptMenu} would not be optimal or the condition $\sum_{o\in A^{v_j}} q_{M'}(o)s_o \geq n_{v_j} + \sum_{v:A^{v_j}\subsetneq A^{v}} n_{v}$ must not hold. This would be a contradiction. 
    
    Consequently, we have that the size of $M'$ must be greater than or equal to the size returned by \cref{alg:LaminarPessMenu}. Consequently, we have that \cref{alg:LaminarPessMenu} returns the minimum possible size of a pessimistically valid menu under $I$.

    \paragraph{Running Time.} \cref{alg:LaminarPessMenu} invokes $\textsc{Ident}$ (\cref{alg:IdentDP}) $|\mathcal{L}|\leq n-1$ times and \textsc{Ident} takes time  $O(n^3m^2)$. Consequently, \cref{alg:LaminarPessMenu} runs in time $O(n^4m^2)$.
\end{proof}

\section{Omitted Material from \cref{sec:wop}}\label{app:wop}


\begin{proof}
    Fix a \Instance{} instance $I=\langle N,O,(s_o)_{o\in O},(A_i)_{i\in N}\rangle$. 
    
    \paragraph{Laminar.} We first consider the laminar case. Given instance $I$, we shall create an alternate instance $I'$, also with laminar acceptability relations with the same number of agents  and minimal acceptability sets, but which will always satisfy $WoP(I)\leq WoP(I')$ and then reason about its maximum possible waste of pessimism value. We will prove this bound as follows:

    \begin{enumerate}
        \item[i.]   Create an alternate instance $I'$ with laminar acceptability relations, with the same number of agents and same number of minimal acceptability sets. 
        \item[ii.]  Prove that $WoP(I)\leq WoP(I')$.
        \item[iii.] Find a worst case choice of serving sizes and number of minimal acceptability sets for WoP
        \item[iv.]  Demonstrate an instance where this bound is tight. 
    \end{enumerate}
    
    \noindent \underline{Creating $I'$.} Given $I$, construct instance $I'$ s.t. the laminar containment tree of $I'$ is a star graph that has the same number of leaves as that of $I$. Let $A^1,\cdots, A^k$ be the distinct acceptability sets under $I$ where $A^1$ is the root of the containment tree. Let $\mathcal{L}=\{A^{v_1},\cdots, A^{v_t}\}$ be the leaves of the containment tree in $I$.

    We set the distinct acceptability sets under $I'$ as $A^1\cup {o^*},A^{v_1},\cdots, A^{v_t}$. Here, we add an additional option $o^*$ with $s_{o^*}=1$ to the root of the tree. Further, we set the number of agents with acceptability set $A^v$ as $n_v=1$ if $v\in L$ and $n_1=n-t$ where is $n$ is the number of agents in $I$. 

    Observe that any menu that under $I'$ each agent find the same set acceptable as under $I$ or finds a superset of it acceptable. This means that any menu that  is optimistically valid under $I$ will continue to be optimistically valid under $I'$. Consequently, $\min_{M\in \Lambda(I)} \sum_{o\in O}q_M(o)s_o\geq \min_{M\in \Lambda(I')} \sum_{o\in O}q_M(o)s_o$. 
    
    Further, for pessimistically valid menus, from \cref{lem:LaminarPess}, the number of agents that need to be fed by any minimal set under $I'$ is at least as much as that under $I$. As a result, any menu that is pessimistically valid under $I'$ will be pessimistically valid under $I$. Consequently, $\min_{M\in \Gamma(I)} \sum_{o\in O}q_M(o)s_o\leq \min_{M\in \Gamma(I')} \sum_{o\in O}q_M(o)s_o$. Thus, we have that 
\sloppy
    \[WoP(I)=\frac{\min_{M\in \Gamma(I)}\sum_{o\in O}q_Ms_o}{\min_{M\in \Lambda(I)}\sum_{o\in O}q_Ms_o}\leq \frac{\min_{M\in \Gamma(I)}\sum_{o\in O}q_Ms_o}{\min_{M\in \Lambda(I')}\sum_{o\in O}q_Ms_o}\leq \frac{\min_{M\in \Gamma(I')}\sum_{o\in O}q_Ms_o}{\min_{M\in \Lambda(I')}\sum_{o\in O}q_Ms_o}=WoP(I').\]

    \noindent \underline{Worst Case $WoP(I')$.} We now compute the worst case value for the waste of pessimism for $I'$. Firstly, given an acceptability set $A^v$ if $\alpha=\min_{o\in A^v} s_o$ then the minimum possible menu size to feed $n'$ agents from $A^v$ would be at most $\alpha \lceil \frac{n}{\alpha}\rceil$. That is, in order to feed all $n'$ agents who find $A^v$ acceptable, we can simply order the item with the smallest serving size $\lceil \frac{n'}{\alpha}\rceil$ times in order to find a valid menu. Consequently, the minimum sized pessimistic menu size to feed $n$ agents under $I'$ would involve feeding $n-t+1$ agents from each leaf vertex $A^{v_j}$ for $j\in [t]$. Thus, the size of the smallest pessimistic menu will be at most $\sum_{j\in [t]}\alpha_j\lceil \frac{n-t+1}{\alpha_j} \rceil$ where $\alpha_j=\min_{o\in A^{v_j}} s_o$.

    Now recall that under $I'$ exactly one agent has $A_i=A^{v_j}$ for each $j\in [t]$. To feed this agent under an optimistically valid menu, exactly one option needs to be selected from $A^{v_j}.$ Specifically, this should be the item with the smallest serving size, i.e. it must have size $\alpha_j$. Thus, the minimum sized optimistically valid menu must have size at least $\sum_{j\in [t]} \alpha_j$. If for some $j\in [t]$, $\alpha_j>1$, this can additionally feed some agents at the root. Further, we had added option $o^*$ to $A^1$ with size $s_{o^*}=1$. If $\sum_{j} \alpha_j<n$, we can order $n-\sum_j \alpha_j$ quantities of $o^*$. Consequently, a minimum sized optimistic menu will have size $\max (n,\sum_{j\in [t]}\alpha_j)$. We now give a bound on $WoP(I')$ based on which term comes in the denominator.

    \textbf{Case 1:} $\sum_{j\in [t]}\alpha_j\geq n$. Recall that the size of the minimum sized pessimistic menu is 

    \begin{align*}
        \min_{M\in \Gamma(I')} \sum_{o\in O}q_M(o)s_o &\leq \sum_{j\in [t]} \alpha_j \left\lceil \frac{n-t+1}{\alpha_j} \right\rceil\\ 
                                                      &\leq \left\lceil \frac{n-t+1}{\min_{j\in [t]}\alpha_j} \right\rceil \sum_{j\in [t]} \alpha_j.\\
                                                      \text{This gives us that }
        WoP(I')                                       &= \frac{\left\lceil \frac{n-t+1}{\min_{j\in [t]}\alpha_j} \right\rceil\sum_{j\in [t]} \alpha_j}{\sum_{j\in [t]} \alpha_j}\\
                                                      &\leq \left\lceil \frac{n-t+1}{\min_{j\in [t]}\alpha_j} \right\rceil.
    \end{align*}

     Observe that this bound would be tight only when $\min_j \alpha_j=\alpha_j$. However, we still need $\sum_{j\in [t]}\alpha_j\geq n$. Thus, in order to maximize the waste of pessimism, we need $\alpha_j=\frac{n}{t}$ for all $j\in [t]$. Thus, to maximize $WoP(I')$, we need that $\sum_{j\in [t]}\alpha_j=n$. This brings us to the next case.

     \textbf{Case 2:} $\sum_{j\in [t]}\alpha_j\leq n$. In this case, we have that the denominator in the Waste of Pessimism bound will be $n$. For the numerator we have that

     \begin{align*}
          \min_{M\in \Gamma(I')} \sum_{o\in O}q_M(o)s_o &\leq \sum_{j\in [t]} \alpha_j \left \lceil \frac{n-t+1}{\alpha_j} \right\rceil\\ 
                                                        &\leq \sum_{j\in [t]} (n-t+\alpha_j)\\
                                                        &=t(n-t) + \sum_{j\in [t]} \alpha_j.
     \end{align*} 

     Again, to maximize $WoP(I)$, we set $\sum_j \alpha_j=n$. Consequently, it remains to choose a value of $t$ that maximizes the function $t(n-t)$. By some straightforward calculus,  we get that we need $t=\frac{n}{2}$, in order to maximize $WoP(I')$. Recall that we assume $\sum_{j\in t}\alpha_j=n$. By definition, for each $j\in [t]$, it must be that $\alpha_j\geq 1$. 
     
     For this case, in order to ensure that the upper bound is tight, we require that $\sum_{j\in t}\alpha_j=n$ and $\sum_{j\in [t]} \alpha_j \left\lceil \frac{n-t+1}{\alpha_j} \right\rceil=\sum_{j\in [t]} (n-t+\alpha_j)$. To this end, we need that for every $j\in [t]$, $\alpha_j$ divides $n-t$. This includes the choice of $\alpha_j=\frac{n}{t}$, which gave the tight bound for Case 1, which in this setting gives us $\alpha_j=2$ for all $j$. To simplify our analysis for the upper bound, we assume $\alpha_j=2$.

     To this end, we get that for any \Instance{} instance $I$ with $n$ agents, 

     \[WoP(I)\leq \frac{(n/2)(n/2+2)}{n}=\frac{n+4}{4}.\]

     \noindent \underline{Tight example.}  Consider the following instance where this bound is tight: set $n=4t'$, $O=\{o_1,\cdots, o_{2t'+1}\}$ where $s_{o}=2$ for all $o\in O$. Further, set $A_i=\{o_i\}$ for $t\leq 2t'$ and $A_i=O$ for $i>2t'$. Here, the minimum sized optimistic menu would have to order at least one quantity of $o_i$ to feed agent $i$ for each $i\leq 2t'$. Each $o_i$ can also simultaneously feed agent $2t'+i$ for $i\leq 2t'$. Consequently, we have that the minimum optimistically valid menu has size $n$. 
     
     Meanwhile, in order to ensure pessimistic validity, each of $n/2$ minimal sets need to feed $n/2+1$ people. As $n=4t'$, we need enough of $o_i$ to feed $2t'+1$ agents. As $s_{o_i}=2$, we need to order $t'+1$ quantities of $o_i$, for each $i\in [2t']$. Thus, the minimum pessimistic menu would have size $(\frac{n}{2})(\frac{n}{2}+2)$. This exactly matches the upper bound obtained.

    \paragraph{Chained.} We now consider the case where the laminar tree is simply a path. That is, for the distinct acceptability sets $A^1,\cdots,A^k$, we have that they form a chain of subsets $A^k\subseteq \cdots \subseteq A^1$.

    Observe that under a pessimistically valid menu must be able to feed everyone from $A^k$, while an optimistically valid menu need not. As we argued in the case for laminar, the minimum sized menu to feed $n$ agents from $A^k$ would have size at most $\alpha \lceil \frac{n}{\alpha}\rceil\leq n+\alpha -1$ where $\alpha=\min_{o\in A^k} s_o$. Consequently, in order to maximize the waste of pessimism under $I$, we need $\alpha$ to be suitably large.
    
    If $\alpha\geq n$, clearly, there would be no waste of pessimism, as even in the minimum sized optimistic menu, to feed the agents with acceptability set $A^k$, we need a size of at least $\alpha$ which will allow us to feed all other agents as well. Thus, the maximum possible waste of pessimism would come when $\alpha=n-1$. As a result, we have that

    \[WoP(I)\leq \frac{n+n-1}{n}=\frac{2n-1}{n}\leq 2.\]

    A tight example would come when we have $n$ agents with two options $O=\{o_1,o_2\}$ where $s_{o_1}=n-1$ and $s_{o_2}=1$. We set $A_1=A_2=\cdots=A_{n-1}=\{o_1\}$ and $A_n=\{o_1,o_2\}$. 

    \paragraph{Identical.} Observe that under identical acceptability relations every optimistically valid menu is pessimistically valid. Hence there would be no waste of pessimism. 
\end{proof}

\end{document}